\documentclass[journal]{IEEEtran}

\usepackage{url}
\usepackage{amsfonts}
\usepackage{times}
\usepackage{threeparttable}
\usepackage{pslatex}

\usepackage{graphicx,caption,subcaption}
\usepackage[normalem]{ulem}
\usepackage{xspace}
\usepackage{lipsum}
\usepackage[]{algorithm2e}
\usepackage{amssymb,amsmath,amsthm}
\usepackage{multirow,array}
\usepackage{url}
\usepackage{color} 
\usepackage{epstopdf}
\definecolor{Gray}{gray}{0.2}
\definecolor{Black}{gray}{0.0}
\definecolor{DarkGreen}{rgb}{0,0.5,0}
\usepackage{hyphenat}
\usepackage{enumitem}
\usepackage{booktabs}
\newcommand{\subparagraph}{}
\usepackage{titlesec}

\setlist[itemize]{leftmargin=*}
\setlist[enumerate]{leftmargin=*}

\graphicspath{{../}}

\usepackage{setspace}
\newcommand{\sysn}{\texttt{Concise}\xspace}
\newcommand{\boldsysn}{\textbf{Concise}\xspace}
\newcommand{\sysnintitlelarge}{Concise\xspace}
\newcommand{\sysnintitleLarge}{Concise\xspace}
\newcommand{\dtsn}{Othello\xspace}
\newcommand{\dtsnintitlelarge}{Othello\xspace}
\newcommand{\dtsnintitleLarge}{Othello\xspace}
\newtheorem{theorem}{Theorem}

\newcommand{\qinsays}[2][]{}
\newcommand{\yysays}[1]{\hspace{0pt}}
\newcommand{\dbsays}[1]{}

\newcommand{\cc}{{\mathtt{cc}}}

\newcommand{\abs}[1]{\left| #1 \right|}

\newcommand{\arra}{\boldsymbol{a}}
\newcommand{\arrb}{\boldsymbol{b}}
\newcommand{\othO}{\boldsymbol{O}}

\newcommand{\crunchBeforeTriFigCap}{\vspace{-0.0ex}}
\newcommand{\crunchAfterColFig}{\vspace{0ex}}
\newcommand{\crunchitemize}{\vspace{-0ex}}
\newcommand{\crunchitemizelarge}{\vspace{-0ex}}
\newcommand{\crunchASS}{\vspace{-0ex}}
\newcommand{\crunchBSS}{\vspace{-0ex}}

\newenvironment{TON}{ \color{black} }{ \color{black} }

\begin{document}

\newcommand{\mysubsubsection}[1] {
    \subsubsection{#1}
}
\renewcommand{\paragraph}[1]{

\noindent \textbf{#1}

}
\newcommand{\superscript}[1]{\ensuremath{^{\textrm{#1}}}}
\def\sharedaffiliation{\end{tabular}\newline\begin{tabular}{c}}

\newcommand{\widthinTriColumn}{0.72\linewidth}
\newcommand{\indentinbox}{\noindent \hangafter=1 \setlength{\hangindent}{1em}}

\title{\LARGE \bf  Memory-efficient and Ultra-fast Network Lookup and Forwarding using Othello Hashing}

\author{Ye~Yu, \emph{Student Member, IEEE and ACM,}
        Djamal~Belazzougui,
        Chen~Qian, \emph{Member, IEEE and ACM,}
         and Qin~Zhang%
        \thanks{Ye Yu (ye.yu@uky.edu) is with University of Kentucky.
        Djamal Belazzougui(dbelazzougui@cerist.dz) is with CERIST.
Chen Qian (qian@ucsc.edu) is with University of California Santa Cruz.
 Qin Zhang (qzhangcs@indiana.edu) is with Indiana University Bloomington.  Chen Qian is the corresponding author.

Ye Yu and Chen Qian were   supported by National Science Foundation Grants  CNS-1701681 and CNS-1717948. Qin Zhang was supported in part by NSF CCF-1525024 and IIS-1633215.

         A preliminary version of this paper was published in \emph{Proceedings of IEEE ICNP 2017 \cite{ConciseICNP}}. }
}

\maketitle

\begin{abstract}
Network algorithms always prefer low memory cost and fast packet processing speed. Forwarding information base (FIB), as a typical network processing component, requires a scalable and memory-efficient algorithm to support fast lookups. In this paper, we present a new network algorithm, Othello Hashing, and its application of a FIB design called Concise, which uses very little memory to support ultra-fast lookups of network names. Othello Hashing and Concise make use of minimal perfect hashing and relies on the programmable network framework to support dynamic updates. Our conceptual contribution of Concise is to optimize the memory efficiency and query speed in the data plane and move the relatively complex construction and update components to the resource-rich control plane. We implemented Concise on three platforms. Experimental results show that Concise uses significantly smaller memory to achieve much faster query speed compared to existing solutions of network name lookups.

\end{abstract}

\label{sec:intro}
\hyphenation{Cu-ckoo-Switch Scale-Bricks}
Significant efforts have been devoted to the investigation and
deployment of new network technologies in order to simplify network management and to accommodate emerging network applications. Though different proposals of new network technologies focus on a wide range of issues, one consensus of most new network  designs is the separation of network identifiers and locators \cite{Saltzer}, which are combined in IP addresses in the current Internet. Instead of IP, flat-name or namespace-neutral architectures have been proposed to provide persistent network identifers. A flat or location-independent namespace has no inherent structure and hence imposes no restrictions to referenced elements \cite{layered}.

The Salter's taxonomy of network elements \cite{Saltzer} is one of the early proposals that suggest the separation of network identifiers and locators.
We summarize an (incomplete) list of reasons of using flat or location-independent names in proposed network architectures:
\begin{itemize}[leftmargin=*]
\crunchitemize
\item To simplify network management, pure layer-two Ethernet is suggested to interconnect large-scale enterprise and data center networks\mbox{\cite{Seattle, Greenberg2009,PAST}}, where MAC addresses are identifiers.
\crunchitemize
\item Software Defined Networking (SDN)  uses matching of multiple fields in packet header space to perform fine-grained  per-flow control. Flow IDs can also be considered names, though they are not fully flat.
\crunchitemize
\item Flat network identifiers have been suggested by various works to support host mobility and multi-homing, including HIP~\cite{HIP},  Layered Naming Architecture~\cite{layered},
    and MobilityFirst \cite{MobilityFirst}.
\crunchitemize
\item AIP \cite{AIP}  applies flexible addressing to ensure trustworthy communication.
\crunchitemize
\item The core network of Long-Term Evolution (LTE) needs to forward downstream traffic according to the Tunnel End Point Identifier (TEID) of the flows \cite{ScaleBricks}.
\crunchitemize
\end{itemize}

The most critical problem caused by location";independent names is  \emph{Forwarding Information Base (FIB) explosion}. A FIB is a data structure, typically a table, that is used to determine the proper  forwarding actions for  packets, at the data plane of a forwarding device (e.g, switch or router). Forwarding actions include sending a packet to a particular outgoing interface and dropping the packet.
Determining proper forwarding actions of the names in a FIB is called name switching.
Unlike IP addresses, location-independent names are difficult to aggregate due to the lack of hierarchy and semantics. The increasing population of network hosts results in huge FIBs and their continuing fast growth.

\newcolumntype{C}[1]{>{\centering\let\newline\\\arraybackslash\hspace{0pt}}m{#1}}
\begin{table*}[ht]
 \centering \small
  \begin{tabular}{cC{2.0cm}C{2.3cm}C{1.3cm}C{7.8cm}}
     \toprule
FIB     & Construction Time & Query Structure Size (bits) &Query Time&Note \\
     \hline
     \sysn  & $O(n)$ & $\leq 4n \log w$ & $O(1)$  & Exact 2 memory reads per query. \\
     \hline
     \texttt{(2,4)-}Cuckoo \cite{CuckooSwitch} & $O(n)$ & $\sim 1.1 n(L\!+\!\log w) $ & $O(1)$  & Up to 8 memory reads per query. \\
     \hline
     SetSep \cite{ScaleBricks} & $O(n\log w)$ & $(2\!+\!1.5\log w)n $& $O(\log w)$  & No method for updates. Not designed as FIB in \cite{ScaleBricks}.\\
     \hline
     BUFFALO \cite{buffalo} & $O(n t)$ & $\alpha n$ & $O(t w)$  & \small{Probabilistic results. False positive ratio affected by $t$ and $\alpha$.} \\
     \hline
     TSS \cite{Srinivasan1999}& $O(n(t\!+\!\log w))$ & $O(n(t\!+\!\log w))$ & $O(t)$  & Designed\,for\,names\,with\,$t$\,fields. $t=O(L)$. \\
\bottomrule
\end{tabular}
  \small\caption{Comparison among FIBs. $n$: \# of names. $L$: length of names. $w$: \# of possible actions. \textbf{In practice, \boldsysn achieves 7\% to 40\% memory and $>$2x speed compared to Cuckoo, though they share the same order of big O time complexity}. }
\label{tbl:comp}
\vspace{-3.7ex}
\end{table*}

On the other hand, the increasing line speed requires the capability of fast forwarding. To support multiple 10Gb Ethernet links, a FIB may need to perform hundreds of millions of lookups per second.
Existing high-end switch fabrics use fast memory, such as TCAM or SRAM, to support intensive FIB query requests.
However, as discussed in many studies
\cite{buffalo,payless,DIFANE}, fast memory is expensive, power-hungry, and hence very limited on forwarding devices.
Therefore, achieving \emph{fast queries} with \emph{memory-efficient} FIBs is crucial for the new network architectures that rely on \emph{location";independent names}. If FIBs are small and increase very little with network size, network operators can use relatively inexpensive switches to build large networks and do not need frequent switch upgrade when the network grows. Hence, the cost of network construction and maintenance can be significantly reduced. For software switches, small FIBs are also important to fit into fast memory such as cache.

In this paper, we present a new FIB design called \sysn. It has the following properties.
\begin{enumerate}
\crunchitemizelarge
\item Compared to existing FIB designs for name switching, \sysn supports \textit{much faster name lookup} using \textit{significantly smaller memory}, shown by both theoretical analysis and empirical studies.
\crunchitemizelarge
\item \sysn can be efficiently updated to reflect network dynamics. A single CPU core is able to perform  millions of network updates per second.  \sysn makes the control plane highly scalable.
\crunchitemizelarge
\item \sysn guarantees to return the correct forwarding actions for valid names. It is \emph{not} probabilistic like those using Bloom filters \cite{buffalo,Caesar}.
\crunchitemizelarge
\end{enumerate}

\sysn is built on a new network algorithm named \dtsn Hashing. \dtsn was inspired by the techniques used in perfect hashing
\cite{MWHC1996,MonotoneHash}.
Different from the static solutions such as Bloomier Filter~\mbox{\cite{BloomierFilter}}, \emph{our unique contribution on \dtsn is to utilize the programmable networking techniques to support network dynamics and corresponding updates}.
\dtsn Hashing and \sysn FIB support fast query and update (addition/deletion of names).
In the resource-limited switches (data plane), \sysn only includes the query component and is optimized for memory efficiency and query speed.
The construction and update components are moved to the resource-rich control plane.
\sysn is constructed and updated in
the control plane and transmitted to the data plane via a standard API such as OpenFlow.
It is the first work to implement minimal perfect hashing schemes to network applications with update functionalities.
\sysn is designed for  name switching, so it does \emph{not} support IP prefix matching.

\sysn is \textbf{a portable solution}, and
it can be used in either software or hardware switches.
We have implemented \sysn in three different computing environments: memory mode, CLICK Modular Router \cite{click}, and Intel Data Plane Developement Kit \cite{DPDK}. The experiments conducted on an ordinary commodity desktop computer show that \sysn uses only few MBs of memory to support hundreds of millions lookups per second, when there are millions of names.

The rest of this paper is organized as follows.
Sec.~\ref{sec:related} presents related work.
We introduce the overview of \sysn in Sec.~\ref{sec:overview}.
We present the \dtsn data structure in Sec.~\ref{sec:othello} and the system design in Sec.~\ref{sec:FIB}.
We present the system implementation and experimental results in Sec.~\ref{sec:evaluation}.
Sec.~\ref{sec:dis} discusses a few related issues.
Finally, we  conclude this work in Sec.~\ref{sec:conclusion}.

\section{Related Work}
\label{sec:related}
\begin{TON}

\textbf{Location-independent network names.} Separating network location from identity has been proposed and kept repeating for over two decades. Numerous network architectures appear in the literature that suggest this concept.
As discussed in Sec. \ref{sec:intro}, a number of new network architectures adopt location-independent names.
A location-independent name can be a MAC address, a tuple consisting of several packet header fields \cite{headerspace1}, a file name \cite{CCN,NDN}, a TEID \cite{ScaleBricks},\,etc.\,To route packets for flat names, ROFL \cite{ROFL} and Disco \cite{Disco} propose to use compact routing to achieve scalability and low routing stretch. ROME \cite{ROME} is a routing protocol for layer-two networks that uses greedy routing whose routing table size is independent of network size.
\sysn is a forwarding structure and does not deal with routing.

\end{TON}

\textbf{FIB scalability.}
We name some techniques used for FIBs and compare them in Table \ref{tbl:comp}.

Hashing is a typical approach to reduce the memory cost of FIBs for name-based switching.
CuckooSwitch \cite{CuckooSwitch}   uses carefully revised Cuckoo hash tables \cite{CuckooHashing} to reach desirable performance on specific high-end hardware platforms.
ScaleBricks \cite{ScaleBricks} also makes use of a memory-efficient data structure \textit{SetSep} to partition a FIB to different nodes in a cluster, it does not store the names as well. We provide a comprehensive comparison of Cuckoo hashing, and \sysn in Sec. \ref{sec:comparison}.
The use of Bloom filters has been proposed in some designs such as BUFFALO \cite{buffalo, Caesar}. However, they may forward packets incorrectly due to the false positives in Bloom filters, causing forwarding loops and bandwidth waste. 
For IP lookups, SAIL \cite{SAIL} and Portire \cite{Asai2015} demonstrate desirable throughput for IPv4 FIB queries. These solutions are usually based on hierarchical tree structures, and their performance are challenged by FIBs with large number of flat names.
The Tuple Space Search algorithm\,(TSS) \cite{Srinivasan1999} is widely used for name matching with multiple files, such as in OpenVswitch and PIECES \cite{Shahbaz2016}.
It is not designed for flat-name switching. Other solutions use hardware to accelerate name switching. For example, Wang\,\emph{et\,al.} \cite{Wang2013} uses GPU to accelerate name lookup in Named Data Networks.

\textbf{Minimal perfect hashing.}
The data structure used in this work, \dtsn, is inspired by the studies on minimal perfect hashing.
    In particular, MWHC \cite{MWHC1996} is able to generate order-preserving minimal perfect hash functions using a random graph.
    MWHC is also presented as Bloomier Filter in \mbox{\cite{BloomierFilter}}.
    The differences between \dtsn and these studies include:
    (1) Both MWHC and Bloomier Filter are designed for static scenarios and they do not support frequent updates like \dtsn does.
    (2) \dtsn uses a bipartite graph instead of a general graph. This design allows much simpler concurrency control mechanism.
    (3) \dtsn is optimized for real network conditions. It performs different functionalities on the control plane and the data plane.
    \dtsn aims to support fast flat name switching, while MWHC is for finding minimum perfect hash functions~\mbox{\cite{MWHC1996}} and Bloomier Filter is designed for approximate evaluation queries~\mbox{\cite{BloomierFilter}}.

\section{Design Overview}
\label{sec:overview}

Consider a network of $n$ hosts identified by unique names. The hosts are connected by SDN-enabled switches.
A logically central controller is responsible of deciding the routing paths of packets. Each switch includes a FIB. The controller communicates with each switch to install and update the FIB.

Each packet header includes the name of the destination host, denoted as $k$. Upon receiving a packet, the switch decides the forwarding action of the packet, such as forward to a port or drop.
We assume the controller knows the set $S$ of all names in the network. In addition, \sysn only accepts queries of valid names, i.e., $k \in S$. We assume that firewalls or similar network functions are installed at ingress switches to filter packets whose destination names do not exist. More discussion about eliminating invalid names is presented in Sec.~\ref{sec:filter}.

\sysn makes use of a data structure \dtsn. 
\dtsn exists in both the switches (data plane) and the controller (control plane).
It has two different structures in the data plane and control plane:
\begin{itemize}
        \crunchitemize
\item \textbf{\dtsn query structure} implemented in a switch is the FIB. It only performs name queries. The memory efficiency and query speed is optimized and the update component is removed.
    \crunchitemize
\item \textbf{\dtsn control structure} implemented in the controller maintains the FIB as well as other information used for FIB construction and updates, such as the routing information base (RIB).
\end{itemize}
Upon network dynamics, the control structure computes the updated FIBs of the affected switches. The modification is then sent from the controller to each switch.

Separating the query and control structures is a perfect match to the programmable networks such as SDN. We call this new data structure design as a \textbf{Polymorphic Data Structure} (PDS). PDS is the key reason that we can apply minimal perfect hashing techniques in programable networks. 
PDS differs from the current SDN model. SDN separates the RIB and FIB to the control and data plane
respectively. We further move part of the FIB to the control plane to minimize the data plane resource cost.

\begin{figure}[t]
    \centering
    \includegraphics[width=0.88\linewidth]{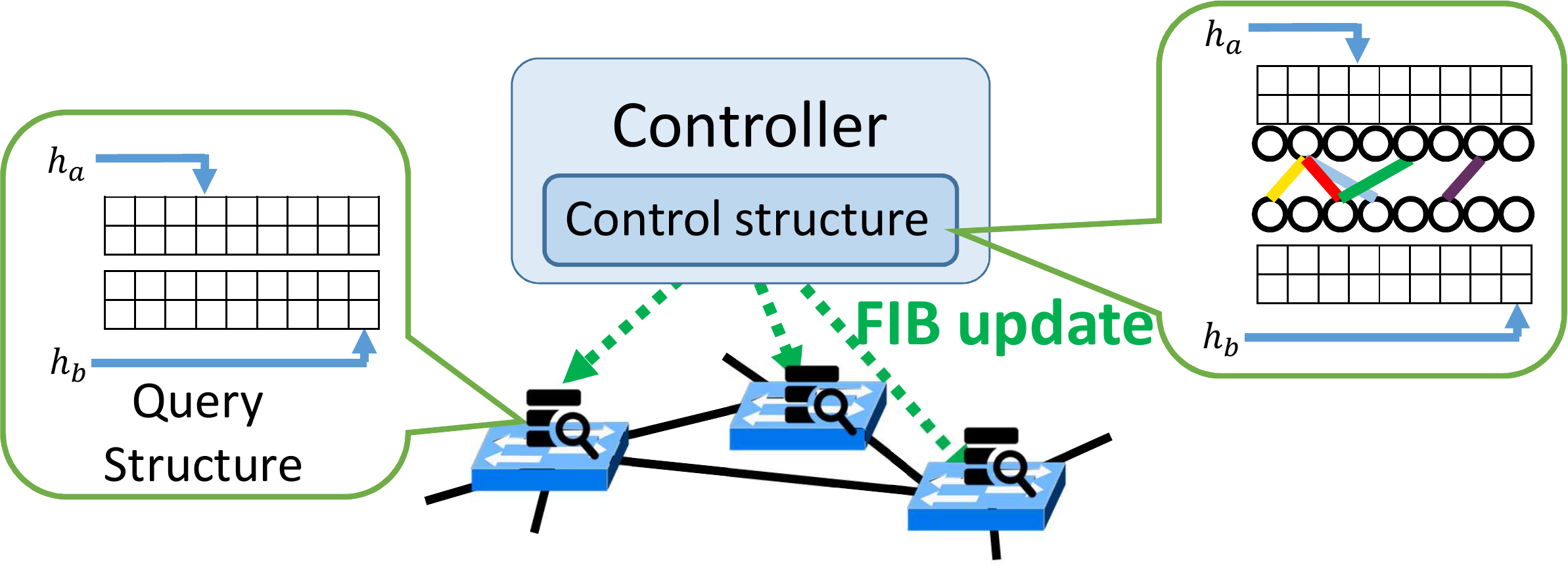}
    \caption{Network Overview of \sysn}
    \crunchAfterColFig
\end{figure}

\section{\dtsnintitleLarge Hashing}
\label{sec:othello}

In this section, we describe the \dtsn data structure.
Inspired by the MWHC minimal perfect hashing algorithm \cite{MWHC1996},
we design \dtsn specially for maintaining the FIB.  
The Bloomier filter \cite{BloomierFilter} can be considered as a special case of the static version of \dtsn.

The basic function of a FIB is to classify all names into multiple sets, each of which represents a forwarding action.
Let $S$ be the set of all names. $n = |S|$.
An \dtsn classifies $n$ names into  two disjoint sets $X$ and $Y$: $X \cup Y = S$ and $X \cap Y = \varnothing$.
 \dtsn can be extended to classify names into $d$ ($d>2$) disjoint sets, serving as a FIB with $d$ actions.

\subsection{Definitions}
\label{sec:othdef}
An \dtsn  is a seven-tuple  $\langle m_a,m_b, h_a, h_b, \arra, \arrb,G \rangle$, defined as follows.

\vspace{0.3ex}
\fbox{
  \begin{minipage}[h!]{0.95\linewidth}
   \indentinbox $\bullet$ Integers $m_a$ and $m_b$, describing  the size of \dtsn. 

  \indentinbox $\bullet$ A pair of uniform random hash functions $\langle h_a, h_b \rangle$, mapping names to integer values $\{0,1,\cdots,m_a\!-\!1\}$ and $\{0,1,\cdots,m_b\!-\!1\}$, respectively.

\indentinbox $\bullet$  Bitmaps $\arra$ and $\arrb$. The lengths are $m_a$ and $m_b$ respectively.

\indentinbox $\bullet$ A bipartite graph $G$. During \dtsn construction and update, $G$ is used to determine the values in $\arra$ and $\arrb$.
 \end{minipage}
}
\\

Figure \ref{fig:examplestructure} shows an  \dtsn  example.
We require that $m_a = \Theta(n)$, $m_b = \Theta(n)$, and $m_am_b > n^2$.
We provide two options to determine the values $m_a$ and $m_b$. 1) $m_a$ is the smallest power of $2$ such that $m_a \geq 1.33n$ and $m_b=m_a$. 2) $m_a$ is the smallest power of $2$ such that $m_a \geq 1.33n$ and $m_b$ is the smallest power of $2$ such that $m_b \geq n$.  A user may choose either option.
The difference is that for Option 1 we establish a rigorous proof of constant update time and for Option 2 we establish the proof with a constraint on $n$. However Option 2 provides slightly better empirical results.

\dtsn supports the query operation. For a name $k$, it computes $\tau(k) \in \{ 0, 1 \}$. If $k\in X$, $\tau(k)=0$. If $k\in Y$, $\tau(k)= 1$. If $k\notin S$, $\tau(k)$ returns 0 or 1 arbitrarily. The values of $\arra$ and $\arrb$ are determined during \dtsn construction, so that $\tau(k)$ can be computed by:
$$\tau(k) = \arra[h_a(k)] \oplus \arrb[h_b(k)]$$
Here, $\oplus$ is the \emph{exclusive or} (\texttt{XOR}) operation. In other words,
if $k \in X$, $\arra[h_a(k)]\!=\!\arrb[h_b(k)]$;
if $k \in Y$, $\arra[h_a(k)]\!\neq\!\arrb[h_b(k)]$.

\subsection{\dtsn Operations}
\dtsn is maintained via the following operations.

\fbox{
\centering
  \begin{minipage}[h!]{0.93\linewidth}
  \indentinbox $\bullet \ \mathtt{construct}(X,Y)$: Construct an \dtsn  for two name sets $X$ and $Y$.

   \indentinbox $\bullet \ \mathtt{addX}(k)$ and $\mathtt{addY}(k)$: add a new name $k$ into the set $X$ or $Y$.

   \indentinbox $\bullet \ \mathtt{alter}(k)$: For a name $k \in X \cup Y$, move $k$ from set $X$ to $Y$ or from $Y$ to $X$. After this operation, the query result $\tau(k)$ is changed.

  \indentinbox  $\bullet \ \mathtt{delete}(k)$: For a name $k \in X \cup Y$, remove $k$ from set $X$ or $Y$.

        \end{minipage}
}
\subsubsection{Construction}
\label{sec:construct}
 \begin{figure}[t]
 \center
 \includegraphics[width=0.85\linewidth]{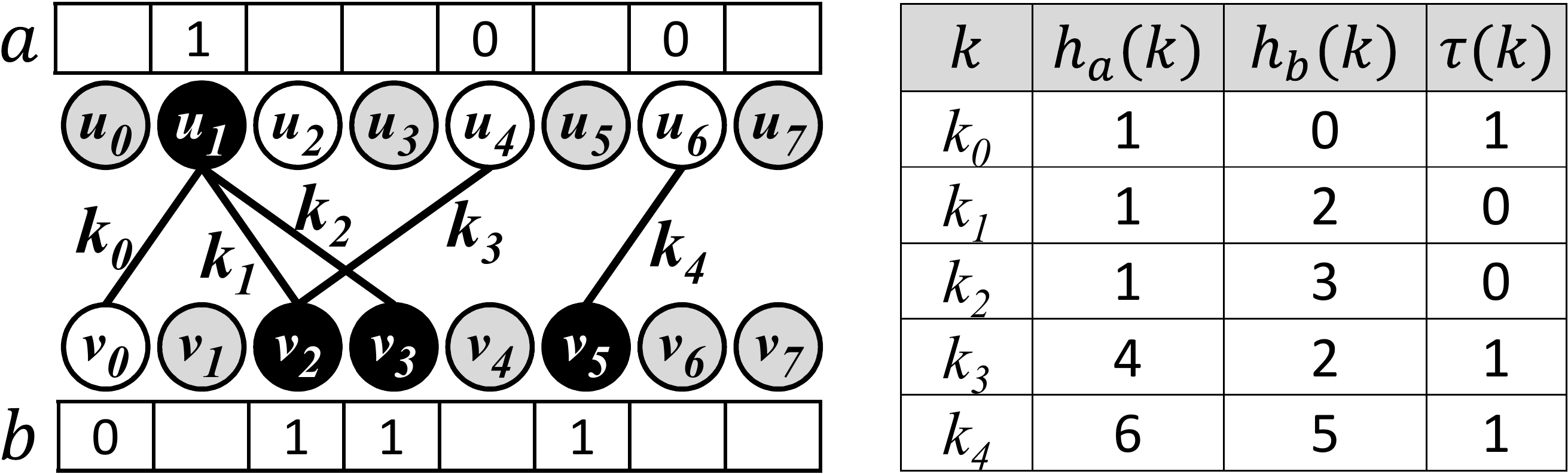}
 \caption{ Example of \dtsn of $n=5$ names with $m_a=m_b=8$. Left: Bipartite graph $G$ and bitmaps $\arra$ and $\arrb$.
 Right: five names $k_0,k_3,k_4 \in X$ and $k_1,k_2 \in Y$; the hash values and $\tau(k)$ values.}
 \label{fig:examplestructure}
 \crunchAfterColFig
 \end{figure}
The $\mathtt{construct}$ operation takes as input two sets of names $X$ and $Y$. The output is an \dtsn $\othO = \langle m_a,m_b,h_a,h_b,\arra,\arrb,G\rangle$.

Here, $G$ is used to determine the hash function pair and the values of $\arra$ and $\arrb$.
$G = (U,V,E)$.
$|U|=m_a$, $ |V| = m_b$.
A vertex $u_i \in U $
or $v_j \in V$
corresponds to bit $\arra[i]$ or $\arrb[j]$.
Each edge in $E$ represents a name.
 There is an edge $(u_i,v_j)\in E$ if and only if there is a name $k \in S$
 such that $h_a(k)=i$ and $h_b(k)=j$.

For each vertex that is associated with at least one edge, the corresponding bit is set to 0 or 1. A vertex associated with bit 0 is colored in white and a vertex associated with bit 1 is colored in black.
For vertices that have no associated edges, the value of the corresponding bits can be set to 0 or 1 arbitrarily,
because they do not affect any $\tau(k)$ value for $k \in S$.
In order to assign correct values of $\arra$ and $\arrb$, \dtsn requires $G$ to be \emph{acyclic}.

The construction algorithm consists of two phases.

\textbf{Phase I: Selecting the hash function pair.}

In this phase, \dtsn finds a hash function pair $\langle h_a,h_b \rangle$.
We assume there are many candidate hash functions and will discuss the implementation in Sec. \ref{sec:hashselection}.
In each round, two hash functions are chosen randomly and $G$ is accordingly generated.
We use Depth-First-Search (DFS) on $G$ to test whether it includes a cycle, which takes $O(n)$ time.
The order in which the edges are visited during the DFS, i.e, the DFS order of the edges is recorded to prepare for the second phase.
Note that if two or more names generate edges with the same two endpoints, we will consider as if there is a cycle.
If $G$ is cyclic, the algorithm will select another pair of hash functions until an acyclic $G$ is found.

\textbf{Phase II:  Computing the bitmaps.}

In this phase, we assign  values for the two bitmaps $\arra$ and $\arrb$.
First, the values in $\arra$ and $\arrb$ are marked as undefined.
Then, we execute the followings for each $e = (u_i, v_j)$ in the DFS order of the edges:
Let $k$ be the name that generates $e$.
If none of $\arra[i]$ and $\arrb[j]$ has been assigned, let $\arra[i] \gets 0$ and  $\arrb[j] \gets \tau(k)$.
If there is only one of $\arra[i]$ and $\arrb[j]$ has been assigned, we can always assign an appropriate value to the other one, such that $\arra[i] \oplus \arrb[j] = \tau(k)$.
As $G$ is acyclic, following the DFS order, we will never see an edge such that both $\arra[i]$ and $\arrb[j]$ have values.

\begin{TON}
We show the pseudocode of \dtsn construction in Algorithm \ref{algo:gen}. 
\begin{algorithm}[ht]
\setstretch{0.93}
\DontPrintSemicolon
\SetKwInput{Input}{Input}
\SetKwInput{Output}{Output}
\Input{ Key-set $X$,$Y$.} 

\Output{An \dtsn structure $\langle m, h_a, h_b, a, b ,G \rangle$ }
\SetKwInput{Proc}{Procedure}
\Begin{
\nl $S \gets X \cup Y$. \;
\nl select $m$ value according to $n = |S|$. \;
\tcc*[f]{Phase I: decide hash function pair}\;
\nl \Repeat{
 $\mathtt{GeneratedGraphIsAcyclic}(S,h_a,h_b)$.
}{
\nl Randomly select hash function $h_a$, $h_b$.
}
\tcc*[f]{Phase II: Compute bitmaps}\;
\nl Compute $G = (U,V,E)$  using $h_a$, $h_b$ and $S$. \;
\nl Execute Depth-First-Search on $G$. \;
\nl    $(e_1,e_2,\cdots,e_n) \gets$ the DFS order of $E$. \;
\nl    Mark all $\arra[i],\arrb[j] ( 0 \leq i,j < m) $ as  $\mathtt{unassigned}$.\;
\nl \For { $t = 1 , 2, \cdots, n$} {
\nl    $k \gets $ the corresponding address for $e_t$.  \;
\nl \lIf {$k \in X$} {$v \gets 0$ \textbf{else} $v \gets 1$.}
\nl    $i \gets h_a(k)$; $j \gets h_b(k)$. \;
\nl    \uIf{\textnormal{both} $\arra[i]$ \textnormal{and} $\arrb[j]$ \textnormal{are} $\mathtt{unassigned}$}{
\nl        $\arra[i] \gets 0$; $\arrb[j] \gets v$.\;
    }
\nl    \uElseIf{$\arra[i]$ \textnormal{\ is\ } $\mathtt{unassigned}$}{
\nl        $\arra[i] \gets \arrb[j] \oplus v$.\;
    }
\nl    \Else (\tcc*[f]{$\arrb[j]$ is$ \mathtt{unassigned}$    })
    {\nl$\arrb[j] \gets \arra[i] \oplus v$.\; }}}
\caption{\dtsn $\mathtt{construct}$ procedure
\label{algo:gen}}
\end{algorithm}

Note that the edges of $G$ are only determined by $S=X\cup Y$
 and the hash function pair $\langle h_a , h_b \rangle$. If we find $G$ to be  cyclic for a given $S$ and a pair $\langle h_a,h_b \rangle$, we shall use another pair  $\langle h_a , h_b \rangle$ to make $G$ acyclic. We show that for a randomly selected pair of hash functions $\langle h_a,h_b\rangle$, the probability of $G$ to be acyclic is very high:
 \begin{theorem}
\label{thm:acyclic}
Given set of names $S = X \cup Y$, 
$n=|S|$.
Suppose $h_a, h_b$ are randomly selected from a family of fully random hash functions. 
$h_a: S\to\{0,1,\cdots,m_a-1\}$,
$h_b: S\to\{0,1,\cdots,m_b-1\}$.
Then the generated bipartite graph $G$
is acyclic with probability $\sqrt{1-c^2}$ when $n$$\to$$\infty$, where $c=\frac{n}{\sqrt{m_am_b}}$, $c<1$.
\end{theorem}

When $G$ is acyclic, we say that $\langle h_a, h_b \rangle $ is a \textit{valid hash function pair} for $S$. We prove Theorem \ref{thm:acyclic} using the technique described in \cite{Botelho2012}.

\begin{proof}
Let $G=(U,V,E)$ be a bipartite random graph with $|U| = m_a$, $|V| = m_b$, $|E| = n$, where each edge is independently taken at random with probability $\frac{n}{m_am_b}$. 
Let $\mathcal{C}_{2\ell}$ be the set of cycles of length $2\ell$ ($\ell\geq 1$) in the complete bipartite graph $K_{m_a},{m_b}$. A cycle in $\mathcal{C}_{2\ell}$ is a sequence of $2\ell$ distinct vertices chosen from $U$ and $V$. Hence, 
$$|\mathcal{C}_{2\ell}| = \frac{1}{2\ell} (m_a)_\ell (m_b)_\ell,$$ where $(m)_\ell = m(m-1)\cdots(m-\ell+1)$. 
Meanwhile, As each edge in $G$ is selected independently, each cycle in $\mathcal{C}_{2\ell}$ occurs in $G$ with probability $(\frac{n}{m_am_b})^{2\ell}$.

As proved in \cite{Botelho2012}, the number of cycles of length $2\ell$ in $G$ converges to a Poisson distribution with parameter $\lambda_{2\ell}$. For $n\to \infty$,
\begin{align*}
 \lambda_{2\ell} & = p^{2\ell} |\mathcal{C}_{2\ell}| \\
& = (\frac{n}{m_am_b})^{2\ell} \frac{1}{2\ell} (m_a)_{\ell} (m_b)_\ell \to \frac{1}{2\ell} \frac{n^{2\ell}}{(m_am_b)^\ell}
\end{align*}
Let $ c = \frac{n}{\sqrt{m_am_b}}$ we have $\lambda_{2\ell} \to \frac{1}{2\ell}c^{2\ell}$ as $n \to \infty$.

The number of cycles of any even length in $G$, represented as a random variable $\mathcal{X}$, converges to a Poisson distribution with parameter $\lambda_e$, where 
$$ \lambda_e = \sum_{\ell = 1}^{\infty} \lambda_{2\ell} = - \frac{1}{2} \ln(1-c^2).$$

Therefore, the probability that $G$ contains no cycle is 
$$\text{Pr}(\mathcal{X} =0) = e^{-\lambda_e} = \sqrt{1-c^2}.$$

\end{proof}


When $c \leq 0.75$ (i.e, $n \leq 0.75m$), $\sqrt{1-c^2}\geq 0.66$. 
Hence
the expected number of rounds to find an acyclic $G$ in Phase I is
$\frac{1}{\sqrt{1-c^2}} \leq 1.51 $ when $c < 0.75$. The time complexity is $O(n)$ in each round.
The second phase takes
$O(n)$ time to visit $n$ edges and assign values of $\arra$ and $\arrb$.
Hence, the total expected time of $\mathtt{construct}$ is $O(n)$.

\end{TON}


\subsubsection{Name addition} \label{sec:nameadd}
To $\mathtt{add}$ a name $k$ to $X$ or $Y$,
the graph $G$ and two bitmaps should be changed in order to maintain the correct result $\tau(k)$.

The algorithm first computes the edge $ e = (u,v)$ to be added to $G$ for $k$, $u = u_{h_a(x)}$, $v = v_{h_b(x)}$. Note that
    $G$ can be decomposed into connected components.
As shown in Figure \ref{fig:exampleadd}, $e$ must fall in one of the following cases.

\begin{figure}[t]
    \centering
    \includegraphics[width=0.45\linewidth]{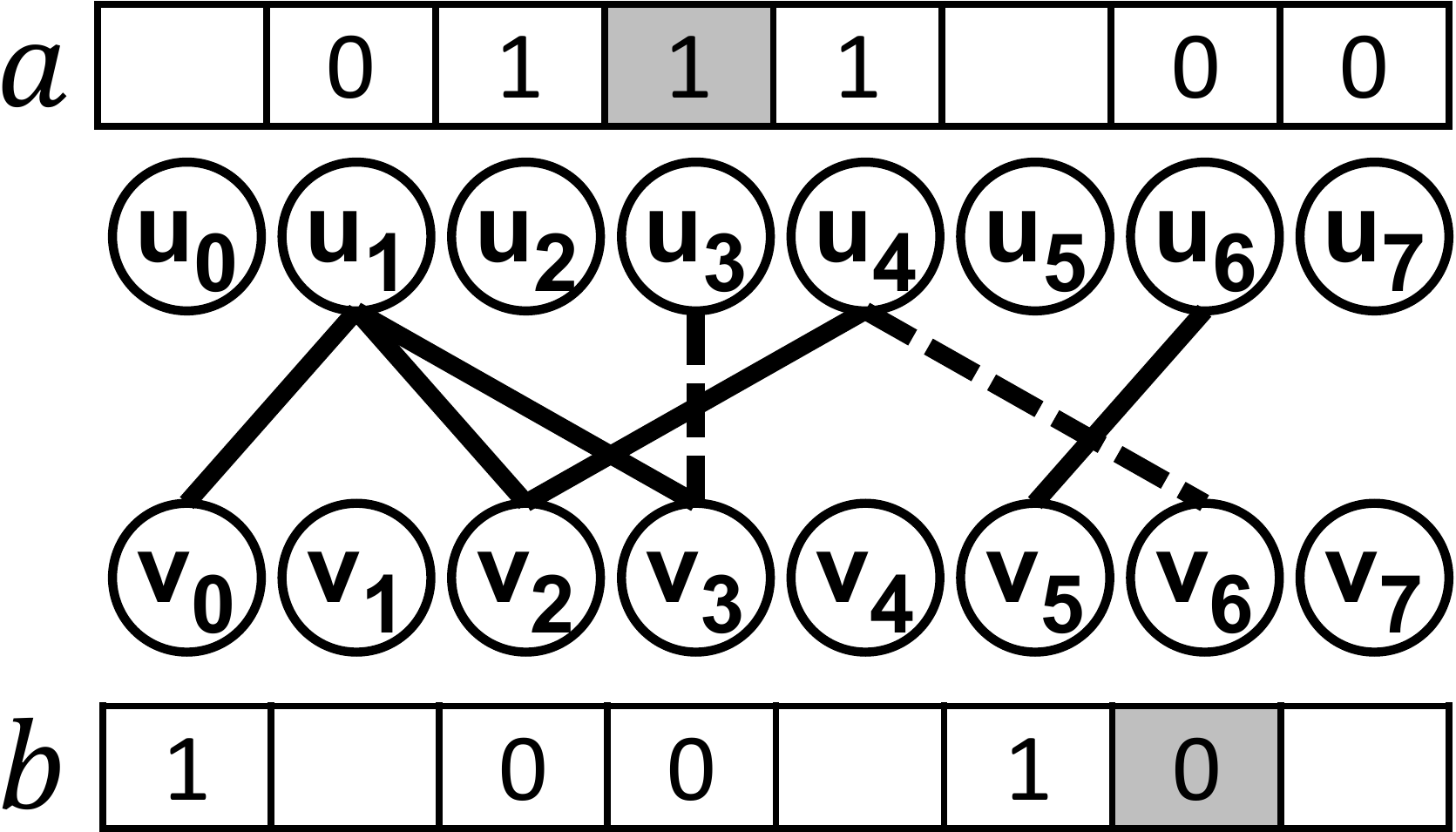}~~
        \includegraphics[width=0.45\linewidth]{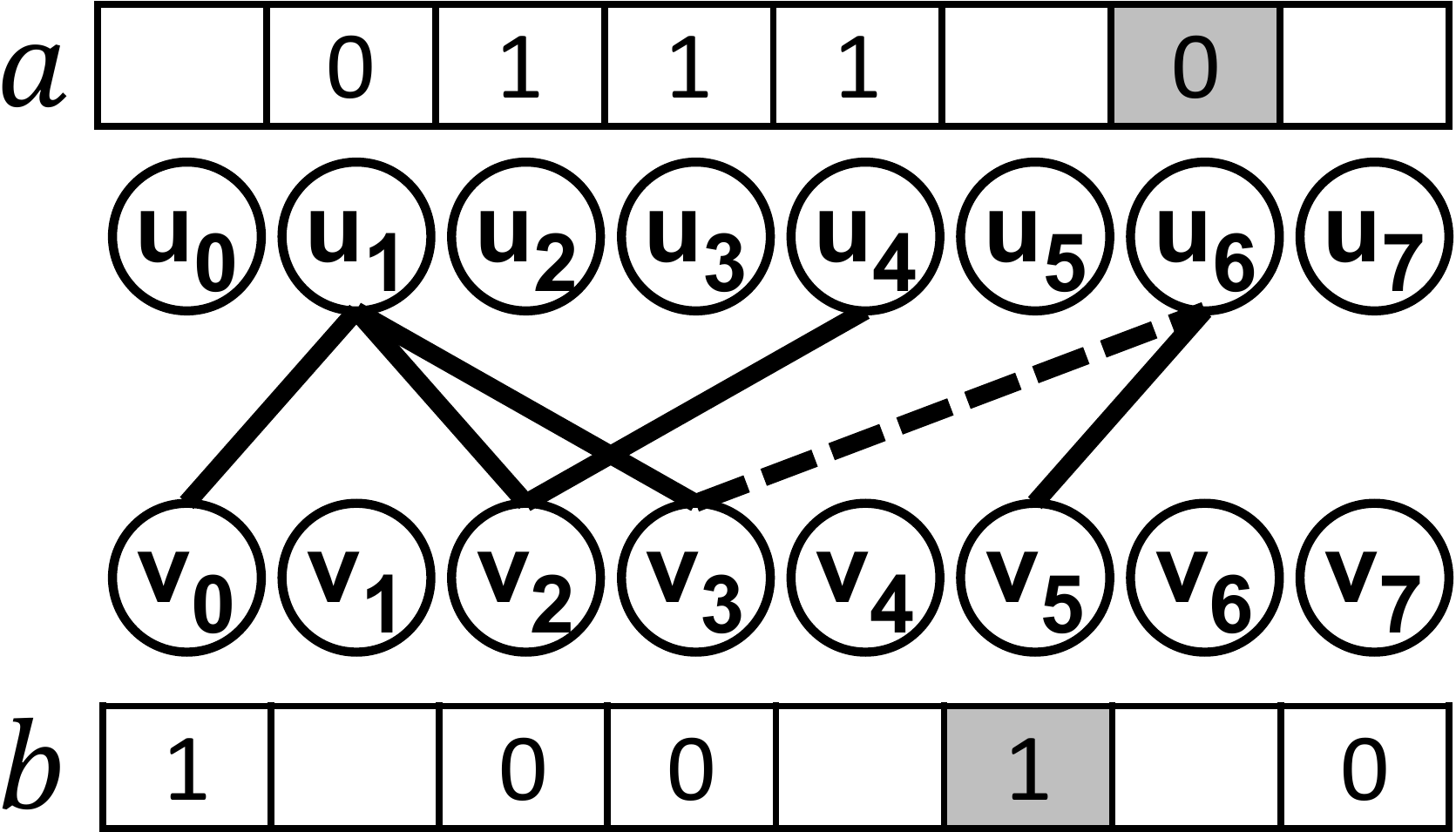}
    \caption{Example of \dtsn update. Dashed edges representing added keys. Gray cells: modified values in $a$ and $b$. Left: Case I, Right: Case II. }
    \label{fig:exampleadd}
    \crunchAfterColFig
\end{figure}

 \emph{Case I:~}$u$ and $v$ belong to the same connected component $\cc$. Adding $e$ to $G$ will introduce
a cycle. In this case, we have to re-select a hash function pair $\langle h_a, h_b \rangle$ until a valid hash function pair is found for the new name set $S \cup \{k\}$. The $\mathtt{construct}$ algorithm is used to perform this process.

\emph{Case II:~}$u$ and $v$ are in two different connected components.
Combining the two connected components and the new edge, we have a single connected component that is still acyclic.
As discussed in Sec. \ref{sec:construct}, it is simple to find a valid coloring plan for an acyclic connected component. Hence, the values of $\arra$ and $\arrb$ can also be set properly.
In fact, at least one of the two connected components can keep the existing value assignments.

\begin{TON}
\textbf{Complexity Analysis.}
We now compute the time complexity of $\mathtt{add}$ using three theorems. 
In particular, we will show that the time complexity of the $\mathtt{add}$ operation is $O(1)$. 
The proof is established by computing the \textit{susceptibility} of graph $G$, namely $\chi(G)$. 
We give a closed-form estimation for $\chi(G) = \frac{1}{1-p}$ where $p= \frac{n(m_a+m_b)}{2m_am_b}$, 
and prove that $\chi(G)$ has a constant upperbound $E[\chi(G)]\leq 4$. 
We are able to compute the closed-form formulae for $\chi(G)$ when $m_a = m_b$. For the case $m_a = 2m_b$, we give a looser upper bound. The numerical estimation shows that the upper bound $E[\chi(G)]\leq 4$ is true for both of the two situations where $m_a = m_b$ and $m_a = 2m_b$.

For the sake of analysis we let ${\mathcal G}_A(m_a,m_b,n)$ be a random acyclic graph generated using the same process as ${\mathcal G}(m_a,m_b,n)$ except that an edge is not added if it introduces a cycle in the graph. It could also be generated by repeatedly generating graphs ${\mathcal G}(m_a,m_b,n)$ until we get an acyclic graph.
It is evident that this random graph model corresponds to the graphs constructed and maintained by Othello.

As stated before, there are two options in choosing values $m_a$ and $m_b$. In Option 1, $m_a=m_b$ and in Option 2, $m_a=m_b$ or $m_a=2m_b$.
For the case $m_a=m_b$ we have
 \begin{theorem}
     \label{thm:sus2}
Suppose \ we \ have \ a \ random \ graph\  ${\mathcal G}_A(m_a,m_b,n)$
where $m_a=m_b$ and we randomly select a node $w$ in ${\mathcal G}_A$. Let $\cc(w)$ be the connected component containing $w$. Then the expected value of $|\cc(w)|$ is $\frac{m_a}{m_a-n}$ as $n \to \infty$.
 \end{theorem}

\begin{proof}
The important parameter that governs the complexity of an insertion
is the \emph{susceptibility} of the graph $G$ which is defined as the expected
size of the connected component that contains a randomly chosen node,
and is denoted by $\chi(G)$.

Let $\chi(G) = E[|\mathtt{cc}(w)|]$ where $w$ is randomly selected from $G$ and $|\cc(w)|$ denotes the number of nodes in $\cc(w)$.
In ~\cite[Lemma 1]{devroye2003cuckoo}, it was proved that for a random sparse graph ${\mathcal G}(m_a,m_a,n)$ with $n$ edges, we have $\chi(G) = \frac{2m_a}{2m_a-2n}$ when $n \to \infty$ given that $n<0.999m_a$.
We will show that the same bound holds for a graph ${\mathcal G}_A(m_a,m_a,n)$. It is well known that the largest connected component in a random graph with $n$ edges and $m$ nodes with $n\leq 0.99\cdot m/2$ has size $O(\log n)$ with probability $1-\frac{1}{n^{10}}$~\cite{devroye2003cuckoo}.

We now generate a graph ${\mathcal G}_A(m_a,m_a,n)$ by generating the edges one by one. If an edge
$(v,w)$ makes the graph cyclic, then we do not add it, but instead put it into a set $S$. Let $E$ be the set of $n$ edges
in the generated acyclic graph $G_1$. Then graph $G_2$ with the set of edges $E\cup S$ will  clearly be a graph ${\mathcal G}(m_a,m_a,n')$ with $n'=n+O(\log^2 n)\leq 0.999 m/2$.
Now we have that $\chi(G_1) \le \chi(G_2)$ and $\chi(G_2)=\frac{2m_a}{2m_a-2n'} \to \frac{2m_a}{2m_a-2n}$ when $n \to \infty$.
\end{proof}

For the case $m_a=2m_b$ we have the following result:
 \begin{theorem}
     \label{thm:sus3}
Suppose \ we \ have \ a \ random \ graph\  ${\mathcal G}_A(m_a,m_b,n)$
where $m_a=2m_b$, $n\leq 0.65m_b$, and that we randomly select a node $w$ in ${\mathcal G}$. Let $\cc(w)$ be the connected component containing $w$. Then the expected value of $|\cc(w)|$ is $O(1)$.
 \end{theorem}

\begin{proof}
Again let $\chi(G) = E[|\mathtt{cc}(w)|]$ where $w$ is randomly selected from $G$.
We generate a graph ${\mathcal G}_A(m_a,m_b,n)$ with $n\leq 0.65m_b$
as follows.  Let $V_a$ with $\abs{V_a} = m_a$ be the set of nodes on the left side, and $V_b$ with $\abs{V_b} = m_b$ be the set of nodes on the right side.  We generate edges one by one from random graph ${\mathcal G}_A(m_a + m_b, n)$, and reject an edge $(v,w)$ if either $(v\in V_a \wedge w\in V_a)$ or $(v \in V_b \wedge w \in V_b)$.  The probability of accepting an edge is thus $\frac{4}{9}$.  We stop the generation when we have finished generating the $n$ edges, and we denote the resulting graph by $G_1$. We let $G_2$
be the graph obtained by adding all the rejected edges back to $G_1$. It is clear that $\chi(G_1) \le \chi(G_2)$.  Moreover, $G_2$ is a random
graph ${\mathcal G}_A(m_a+m_b,n')$ with $n'=\frac{9}{4} n\pm O(\sqrt{n})$ with probability $1 - \frac{1}{n^{10}}$.
According to Theorem 3.3(i) in \cite{janson2008susceptibility}, for a graph $G_3={\mathcal G}(m_a+m_b,n')$:
\begin{eqnarray*}
\chi(G_3) &\le& \frac{m_a + m_b}{m_a + m_b - 2n'} = \frac{3 m_b}{3 m_b - 2 \cdot \frac{9}{4} n} \\
&\le& \frac{3 \frac{n}{0.65}}{3 \frac{n}{0.65} - 2 \cdot \frac{9}{4} n} = O(1).
\end{eqnarray*}
We can use the same argument as in the proof of Theorem~\ref{thm:sus2}
to show that the susceptibility for a graph ${\mathcal G}_A(m_a+m_b,n')$
is the same as for a graph ${\mathcal G}(m_a+m_b,n')$ which concludes the proof.
\end{proof}

The following theorem concludes that the time complexity of \texttt{add} is $O(1)$.
\begin{theorem}
    \label{thm:insertion}
Assuming $h_a, h_b$ are randomly selected from a family of fully random hash functions, an insertion into an \dtsn with $n$ existing names will take constant amortized expected time 
when $m_a=m_b$, or when $m_a=2m_b$ and $n\leq 0.65m_b$.
\end{theorem}

\begin{proof}
    In the algorithm described in Section.~\ref{sec:nameadd},
during an insertion, we have to add an edge that connects a randomly selected node $u\in U$ to another randomly selected node $v \in V$.
We will first bound the amortized expected cost of insertions
that fall in \textit{Case I} and then the induced cost of insertions that fall in \textit{Case II}.
Let $|\cc(w)|$ be the size the connected component that contains node $w$.
Let $|\cc_b(w)|\leq |\cc(w)|$ be the number of nodes in $\cc \cup V$. 

The probability
that node $v$ falls in the same connected component as node $w$ is $\frac{|\cc_b(w)|}{m_b}\leq \frac{|\cc(w)|}{m_b}$
which is the probability of reconstruction. Since the reconstruction
takes expected $O(n)$ time, the amortized expected cost is $\frac{|\cc_b(w)|}{m_a}\cdot O(n)=O(|\cc(w)|)=O(1)$.

For \textit{Case II}, the cost is clearly $O(|\cc(w)|+|\cc(v)|)=O(1)$, since we have to traverse
the connected component that results from merging the two connected components that contain $w$
and $v$.
\end{proof}

Note we have  a rigorous proof for Option 1 but Option 2 provides slightly better empirical results. It is reasonable to conjecture that Theorem~\ref{thm:insertion} also holds for $m_a=2m_b$ without the constraint $n\leq 0.65m_b$.

\begin{figure}[t]
    \centering
    \includegraphics[width=0.75\linewidth]{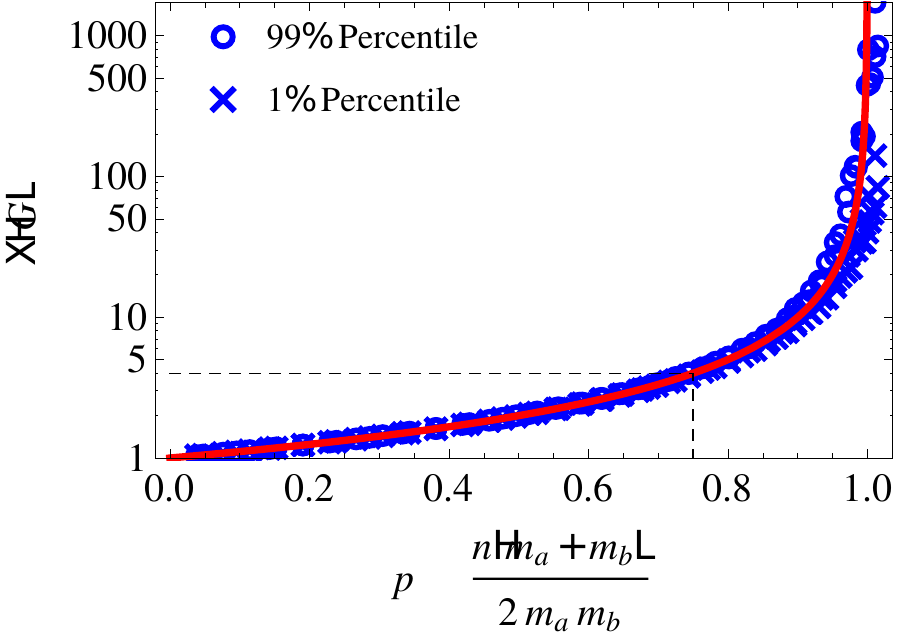}
    \caption{$\chi(G)$ of acyclic graphs vs parameter $p$. Red curve:$ \frac{1}{1-p} $ }
    \label{fig:susbipartite}
\end{figure}

\textit{Numerical estimation of $\chi(G)$}

We conjecture that $\frac{1}{1-p}$, where $p=\frac{n(m_a+m_b)}{2m_am_b}$
is a good estimation for $\chi(G)= E[|\cc(w)|]$,
and present numerical simulation to support our conjecture.
We generate   acyclic bipartite graphs with random $m_a$, $m_b$, and $n$ values (within the range 10K $\sim$ 1M). Then we compute their $\chi(G)$ value.
For a particular $p=\frac{n(m_a+m_b)}{2 m_a m_b}$ value, we randomly sample at least 500 graphs with different $m_a$,$m_b$, and $n$.
In Figure \ref{fig:susbipartite},  we plot the 1-th and 99-th percentile of $\chi(G)$.

As shown in Figure \ref{fig:susbipartite}, when $p$ is not so close to 1, the sampled $\chi(G)$ values are very close to $\frac{1}{1-p}$.
    When $p$ grows larger,
    the sampled $\chi(G)$ values tend to grow slower than $\frac{1}{1-p}$. Hence we conclude that $\frac{1}{1-p}$ is
a good upper bound for $\chi(G)$.
In \dtsn, $\frac{4}{3}n \leq m_a < \frac{8}{3}n$, $n \leq m_b < 2n$. $m_a$ and $m_b$ must be powers of $2$.
For this choice of parameters we can see that $p = \frac{n(m_a+m_b)}{2 m_a m_b} \leq 0.75$
and so $\chi(G) \leq 4$ which is a small constant.

This estimated value $\chi(G)=\frac{1}{1-p}$ is in coherence with the evaluation results on \sysn updates shown in Figure \ref{fig:updateInsert}.

\end{TON}

\textbf{\dtsn size growth.}
After adding a name into \dtsn, $n=|S|$ grows and may violate $m_a \geq 1.33n$ and $m_b\geq n$.
However,  \dtsn  works correctly as long as $G$ is acyclic, even when $m_a < 1.33n$ or $m_b <n$. Hence, \dtsn does not deal with the requirement on $m_a$ and $m_b$ explicitly for additions.
Although the $E[|\cc|]$ value may grow as more names are added to \dtsn, it is always smaller than $10$ in our experiments.
The expected time to add a name to \dtsn is still $O(1)$ in practice.

When adding a new name falling in Case I, the values of $m_a$ and $m_b$ will be updated by $\mathtt{construct}$, which guarantees $m_a \geq 1.33n$ and $m_b \geq n$.

\subsubsection{Set change for a name}

Operation $\mathtt{alter}(k)$ is used to move a name $k$ from set $X$ to set $Y$\,(or from $Y$ to $X$).
The bitmaps $\arra$ and $\arrb$ should be modified so that $\tau(k)$ is changed from $0$ to $1$ (or from 1 to 0).
The graph $G$ does not change during $\mathtt{alter}(k)$.  We only need to change the coloring plan of the connected component that contains the edge $ e = (u_{h_a(k)}, v_{h_b(k)})$. One approach is to ``flip'' the colors of all vertices at one side of $e$, i.e., to change 0 to 1, and to change 1 to 0.
The amortized time cost is $O(1)$.

\subsubsection{Name deletion}
$\mathtt{delete}(k)$ can be done by simply removing the edge $(u_{h_a(k)},v_{h_b(k)})$ in the graph $G$. The bitmaps $\arra$ and $\arrb$ are not modified because the values of $\tau(k)$ after deleting $k$ do not matter anymore. The time complexity  is $O(1)$.

\crunchBSS
\subsection{Query structure and control structure}
\crunchASS
Each \dtsn is  a seven-tuple
$\langle m_a,m_b, h_a, h_b, \arra, \arrb,G \rangle$.
Note that for a query on \dtsn, only the first six elements are necessary for computing the $\tau$ value. The information stored in $G$ is not needed for the query operation. Hence, we let the switches only maintain the six-tuple
$\langle m_a,m_b, h_a, h_b, \arra, \arrb\rangle$ in their local memory, namely the \textit{Query structure}.
Storing this six-tuple takes $2m+O(1)$ bits of memory space. 
The time cost for each query of \dtsn is equal to the sum of the cost of computing two hash values, two memory accesses for the two bitmaps, and one $\mathtt{XOR}$ arithmetic operation.

In comparison, the network controller maintains the seven-tuple, namely the \textit{Control Structure}.
The controller is responsible for maintaining the FIB of the switches in the network. The switches execute the queries on the query structures.

\crunchBSS
\subsection{Summary of \dtsnintitlelarge Properties}
\crunchASS
An \dtsn is decomposed into a query structure running in the data plane and a control structure in the control plane.\,The query structure uses $\leq 4n$ bits for $n$ names. Every query takes a small constant time including computing two hash values and two memory accesses. The control structure uses $O(n)$ bits. The expect time complexity is $O(n)$ for construction and $O(1)$ for name addition, deletion, and set change.
Note that \emph{the distribution of names in $X$ and $Y$ has no impact on the space and time cost of \dtsn}, because $G$ only depends on $S$ and $\langle h_a, h_b \rangle$.
In Sec. \ref{sec:POGidea}, we demonstrate the extension of \dtsn. It classifies names into $d>2$ disjoint sets, while still requiring small memory and constant query time.

\crunchBSS
\section{System Design of \sysnintitleLarge}
\crunchASS
\label{sec:FIB}
We present how to build \sysn using the \dtsn data structure as follows. 
The design also includes the implementation details of FIB update and concurrency control.

\crunchASS
\subsection{Extension of \dtsnintitlelarge for Network Lookups}
\crunchBSS
\label{sec:POGidea}
The extension of \dtsn to support classification for more than two sets is called a Parallel Othello Group (POG).
An $l$-POG is able to classify names into $2^l$ disjoint sets. It serves as a FIB with $2^l$ forwarding actions.
Let $Z_0, Z_1, \cdots, Z_{2^l-1}$ be the $2^l$ disjoint sets of names.
Let $S = Z_0 \cup Z_1 \cup \cdots \cup Z_{2^l-1}$.
    A query on the $l$-POG for a name $k \in S$ returns an $l$-bit integer $\tau(k)$, indicating the index of the set that contains $k$, i.e., $k \in Z_{\tau(k)}$.

The idea of POG is as follows.
Consider $l$ \dtsn{s} $\othO_1, \othO_2, ..., \othO_l$.
Each $\othO_i$ classifies keys in set $X_i$ and $Y_i$ ($1\leq i \leq l$), where $X_i$ and $Y_i$ satisfies:
$$
X_i = \bigcup_{(j \bmod 2^i)< 2^{i-1}} Z_j \text{;} \qquad
Y_i = \bigcup_{(j \bmod 2^i) \geq 2^{i-1}} Z_j\text{.}
$$
Let $\tau_i(k)$ be the query result of $\othO_i$ for name $k$.
Consider the $l$-bit integer
$((\tau_l(k)\tau_{l-1}(k)\cdots\tau_1(k))_2$.
Note that $\tau_i(k)\!=\!0$ if and only if $k\!\in\!X_i$.
Meanwhile, $Z_{\tau(k)} \subset X_i$ if and only if
$(\tau(k) \mod 2^i) < 2^{i-1}$ (the $i$-th least significant bit of $\tau(k)$ is 0).
Hence, the $i$-th least significant bit of $\tau(k)$ equals to $\tau_i(k)$. i.e,
$$\tau(k) = ((\tau_l(k)\tau_{l-1}(k)\cdots\tau_1(k))_2$$

For each $i$ ($1\leq i\leq l$), $X_i \cup Y_i = S$. i.e., the $l$ {\dtsn}s share the same $S$.
Recall that the edges in $G$ is determined by only $S=X\cup Y$ and $\langle h_a, h_b \rangle$, and $\langle h_a, h_b\rangle$ is decided during $\mathtt{construct}$ by $S$.
The $l$ \dtsn{s} may share the same $\langle h_a, h_b \rangle$ and same edges in $G$. However, the bitmaps in different \dtsn{s} are different.

\textbf{Parallelized execution with bit slicing.}
Each operation of an $l$-POG consists of operations on the $l$ \dtsn{s}.
Using the bit slicing technique, these operations can be executed in parallel. The bit slicing technique is widely used to group executions in parallel~\mbox{\cite{Anand2010}}.
An $l$-POG query structure includes $l,m,h_a,h_b$ and two vectors $A$ and $B$.
  Each of $A$ and $B$ contains $m$  $l$-bit integers.
  Consider all the $i$-th bits of the elements in $A$. These bits can be viewed as a \textit{slice} of the array $A$. The $i$-th slice of $A$ is used to represent bitmap $\arra_i$. The slices of $B$ are defined similarly.
Using this technique, $\tau(k)$ can be computed using one arithmetic operation by: $$\tau(k) = A[h_a(k)] \oplus B[h_b(k)] $$

 When $l$ is not larger than the word size of the platform, each $l$-POG query only requires two memory accesses  for fetching $A[i]$ and $B[j]$.
 The arithmetic operation includes computing the hash functions and the $\mathtt{XOR}$.

    All \dtsn operations can be decomposed into two steps: (1) modifications on $G$, (2) operations on some bits in $\arra$ and $\arrb$.
In an $l$-POG, the $l$ Othellos share the same $G$ and \textit{the first step is only executed once} for \textit{all} $l$ \dtsn{s}.
Hence the bit slicing technique also applies to all other operations of POG.

Therefore, the expected time cost of each name addition, deletion, or set change operation is only $O(1)$, instead of $O(l)$. The time complexity of POG construction is still $O(n)$.

\crunchASS
\subsection{Selection of Hash functions}
\crunchBSS
\label{sec:hashselection}

The hash function pair is critical for system efficiency.
Ideally,  $h_a$ and $h_b$ should be chosen from a family of fully random and uniform hash functions.
Similar to the implementation of CuckooSwitch \cite{CuckooSwitch}, we apply a function $H(k,\mathtt{seed})$ to generate the hashes in our implementation. Here, $H$ is a particular hashing method and  $\mathtt{seed}$ is a 32-bit integer. We let $h_a(k) = H(k,\mathtt{seed}_a)$ and $h_b(k) = H(k,\mathtt{seed}_b)$.
Thus, $\langle h_a, h_b \rangle$ is uniquely determined by a pair of integers $\langle \mathtt{seed}_a, \mathtt{seed}_b\rangle$.

The proper hashing method $H()$ is platform";dependent.
\sysn uses the CRC32c function
for robust and faster hash results, which is then effectively mapped to a $t$-bit integer value where $m_a = 2^t$ or $m_b = 2^t$.
Evaluation shows that CRC32c demonstrates desirable performance in practice.

\begin{TON}

\subsection{FIB Update and Concurrency Control}
\label{sec:concurr}

We assume that there is one logically centralized controller in the network.
Upon network dynamics, the controller computes the POGs for a number of switches and update the query structures in the switches by FIB update messages using a standard SDN API. 
If $m, h_a, h_b$ do not change during the update, an update message only contains a list of elements to be modified in $A$ and $B$.
Otherwise, it contains the full query structure of $l$-POG $\langle m,h_a,h_b,A,B \rangle$.

After receiving a FIB update message, a \sysn switch modifies its
POG query structure.
Instead of locks, \sysn uses simple bit vectors to prevent read-write conflicts in the query structure.
Experimental results show that the concurrency control mechanism
has a negligible impact on the network performance.

While each POG query is computed using two elements in $A$ and $B$, there is a chance of a read-write conflict during the update.
In \sysn, the $\mathtt{query}$ always returns correct result. Such concurrency issue is addressed as follows.

\textbf{Concurrency requirements.}
Let $A, B$ be the two vectors of the query structure before an update and $A', B'$ be the ones after the update.
For a name $k$ that exists in the FIB before and after the update, suppose $i\!=\!h_a(k)$ and $j\!=h_b(k)$. Both $A[i] \oplus B[j]$ and $A'[i] \oplus B'[j]$ are considered as correct actions, although they may be different.
Note that, when $A[i]=A'[i]$,
the values $A'[i] \oplus B[j]$
 and
 $A[i] \oplus B'[j]$
 are both correct query results, no matter how read/write events are ordered.
Inconsistency only happens when
both $A[i]$ and $B[j]$ are changed during the update.

\textbf{Concurrency control design.}

\sysn observes whether the vector $A$ is being modified.
For a query for name $k$,
if an update that affects $A[i]$ is being executed,
\sysn does not execute the query until the update finishes.
\sysn maintains two bit vectors $D_1$ and $D_2$ for concurrency control.
All bits in $D_1$ and $D_2$ are set to $0$ during the initialization.
Each index $i$ ($0\leq i < m$) corresponds to an index $p(i)$ in $D_1$ and $D_2$.
The lengths of $D_1$ and $D_2$ are set to 512 bits and $p(i) = i \mod 512$.

\textbf{Update procedure.}
A pseudocode of the update procedure is described in Algorithm \ref{alg:update}.
Before an update of the POG that will change some elements of $A$,
\sysn flips the corresponding bits in $D_1$, i.e., change 0s to 1s and 1s to 0s.
After the update, it flips the bits with same indexes in $D_2$.
For any index $i$, when \sysn observes $D_1[p(i)]\!\neq\!D_2[p(i)]$,
there must be no ongoing update that affects $A[i]$. 
Note that even if a bit index corresponds to multiple elements that are changed in an update, the bit is only flipped once.

\begin{algorithm}[t]
 \KwData{New value at some indexes in $A$ and $B$: $A[i_1], A[i_2], \cdots$, $B[j_1], B[j_2], \cdots.$}
 \KwResult{Updated \sysn query structure}
  \nl \textit{Affected} $ \gets \emptyset$\;
 \nl \ForEach{ $i  \in \{ i_1, i_2, \cdots \}$}{
  \nl \textit{Affected} $\gets$ \textit{Affected} $ \cup \{ i \mod 512 \}$
 }
 \nl \ForEach{ $i \in $ \textit{Affected}}{
 \nl $ D_1[i] \gets 1 \oplus D_1[i]$
 }
  \nl\tcp{reorder barrier}
\nl Update  $A[i_1], A[i_2], \cdots$, $B[j_1], B[j_2], \cdots.$ \;
\nl  \tcp{reorder barrier}
\nl \ForEach{ $i \in $ \textit{Affected}}{
\nl $ D_2[i] \gets 1 \oplus D_2[i]$
 }
 \caption{Update procedure for \sysn}
 \label{alg:update}
\end{algorithm}

\begin{algorithm}[t]
 \KwData{\sysn query structure and name $k$}
 \KwResult{Query result $\tau(k)$}
\nl $i \gets h_a(k)$\;
\nl $j \gets h_b(k)$\;
\nl $p \gets i \mod 512$\;
\nl \While{true}{
\nl $\boldsymbol{\delta}_2 \gets D_2[p]$\;
\nl   \tcp{reorder barrier}
\nl   $\alpha \gets A[i]$\;
\nl   $\beta \gets B[j]$\;
\nl      \tcp{reorder barrier}
\nl       $\boldsymbol{\delta}_1 \gets D_1[p]$\;
\nl        \If{$\boldsymbol{\delta_2} = \boldsymbol{\delta_1}$}{
\nl \KwRet{$\alpha \oplus \beta$}
        }
 }
 \caption{Query procedure on \sysn}
 \label{alg:query}
\end{algorithm}

\textbf{Query procedure. }
A pseudocode of the query procedure is described in Algorithm \ref{alg:query}
The query procedure for name $k$ includes the following three steps.
(1) Fetch the bit $\boldsymbol{\delta}_2\!=\!D_2[p(i)]$.
(2) Fetch the value of $A[i]$ and $B[j]$.
(3) Fetch $\boldsymbol{\delta_1}\!=\!D_1[p(i)]$. If $\boldsymbol{\delta_2}\!=\!\boldsymbol{\delta_1}$,
compute $A[i] \oplus B[j]$ and return it as the query result.
Otherwise, $\boldsymbol{\delta_2} \neq \boldsymbol{\delta_1}$ and we know that the POG is currently being updated and the update affects $A[i]$.
The query for $k$ will stop and
is put in a later place of the query event queue.
\sysn uses reordering barrier instructions to ensure the execution order in both update and query procedures.

Here, the order of flipping
$D_1[p(i)]$ and $ D_2[p(i)]$
during an update
and the order of getting their values during a query are different.
Any updates that affect $A[i]$ and start during a query must result in
$\boldsymbol{\delta_2} \neq \boldsymbol{\delta_1}$.

The above procedures of update and query
should be executed in the given explicit order.
This can be specified by compiler reorder barriers on strong memory model platforms such as x86\_64, or fence instructions on weak memory model platforms such as ARM.

\end{TON}

\section{Implementation and Evaluation}
\crunchASS
\label{sec:evaluation}

We implement \sysn on three  platforms and conduct extensive experiments to evaluate its performance.
\crunchBSS
\subsection{Implementation Platforms}
\crunchASS
\textbf{1. Memory-mode.} We implement the POG query and control structures,  running on different cores of a desktop computer.
In addition, we use a discrete-event simulator to simulate other data plane functions such as queuing.
The memory-mode experiments are used to compare the performance of the algorithms and data structures. They demonstrate the maximum lookup speed that \sysn is able to achieve on a computing device by eliminating the I/O overhead.

\textbf{2. Click Modular Router} \cite{click} is an architecture for building configurable routers.
We implement an \sysn prototype on Click.
It is able to serve as switch that forwards data packets.

\textbf{3. Intel Data Plane Development Kit (DPDK)}
\cite{DPDK} is widely used in fast data plane designs.
We use a virtualized environment to squeeze both the traffic generator and the forwarding engine on the same physical machine. This prototype is able to serve as a real switch that forwards data packets.

\crunchBSS
\subsection{Methodology}
\crunchASS
We compare \sysn with three approaches for name switching: (1) Cuckoo hashing\,\cite{CuckooHashing} (used in CuckooSwitch~\cite{CuckooSwitch} and ScaleBricks~\cite{ScaleBricks}), (2) BUFFALO~\cite{buffalo},
and (3) Orthogonal Bloom filters.
CuckooSwitch~\cite{CuckooHashing} is optimized for a specific platform with 16 cores and 40 MBs of cache. ScaleBricks~\cite{ScaleBricks} is designed for a high performance server cluster.
We were not able to repeat their experiments on commodity desktop computers.
Instead, we compare \sysn with \texttt{(2,4)}-Cuckoo hashing, which is their FIB, by reusing the code from the public repository of CuckooSwitch.
BUFFALO does not always return correct forwarding actions. The false positive rate is set to at most $0.01$\%.
We also implement a new technique called Orthogonal Bloom filters (OBFs) for comparison.
It uses a Bloom filter to replace an Othello for classification of two sets $X$ and $Y$: all names in $X$ hit the Bloom filter. The false positive rate is also set to at most $0.01\%$. The other design of OBFs is similar to \sysn.

We do not include SetSep \mbox{\cite{SetSepHotOS}} in this section although it shares some similarity to \dtsn. The SetSep work \mbox{\cite{ScaleBricks}} does not include an update method and  was not proposed for FIBs.
Also, there is no explicit update algorithm for SetSep in every work in which it has been used \mbox{\cite{SetSepHotOS}}\mbox{\cite{ScaleBricks}}. Hence, SetSep cannot be directly used for FIBs and it is not suitable to implement SetSep and compare it with other FIB designs. Actually our experiments using a static version of SetSep show that \sysn is faster than SetSep for name lookups.

\subsubsection{Performance metrics}
\textbf{Data plane performance metrics}
are used to characterize the performance of the \sysn query structure in switches.

\emph{Memory cost:} the size of memory needed to store a FIB.

\emph{MCQ:} the \textit{m}aximum number of \textit{C}ache lines transmitted per \textit{Q}uery. During each memory access, a cacheline (usually 256 bits of data in many architectures) is transmitted from memory to the CPU. It is used to characterize the time cost of a query.

\emph{Query throughput:} the number of queries that a FIB is able to process per second.

\emph{Query throughput under update:} the query throughput measured when the FIB is being updated.
It reflects the effectiveness of the concurrency control mechanism.

\emph{Processing delay:} the processing delay of the query structure for a packet.
It reflects the ability of the data plane to process burst traffic.
Such metric is measured using an event-based simulator on real traffic trace.

\textbf{Control plane performance metrics }
characterize the performance of the \sysn control structure in the controller.

\emph{Construction time:}
the time to construct a FIB. Note that, for some networks in which $G$ is shared among all switch FIBs such as Ethernet, not every FIB requires the entire construction time. Once $G$ is determined, it can be reused for all switches.

\emph{Update throughput:} the number of updates that can be processed by the control structure per second. Here, an update may consist in adding a name, deleting a name, or changing the forwarding action of a name.

\subsubsection{Evaluation environment and settings}
\label{sec:settings}

\textbf{LFSR name generator}
In the experiments, a series of query packets with different names were generated and fetched by the FIB.
One straightforward approach is to feed the FIB with a publicly available traffic trace.
However, the time for transmitting the data from the physical memory to the cache is too large compared to the FIB query time.
Hence, to conduct more accurate measurement, we use a
linear feedback shift register (LFSR) to generate the names.
 One LFSR generates about 200M names per second on our platform.
In addition, we provide event-based simulation using real traffic data to study the processing delay on \sysn.

In fact, LFSR gives no favor to \sysn  because the names are generated in a round-robin scenario, which provides the minimum cache hit ratio.
LFSR traffic is actually the \textit{worst} traffic for \sysn.
On the contrary, in denial-of-service attack traffic, the queries concentrate on one or few names, and they always hit the cache. Hence, the query throughput of \sysn in DoS attack traffic may be higher than the value measured with LFSR traffic. We believe the result measured in LFSR traffic reflects the true performance of \sysn.

\textbf{Evaluation Settings}
In the following section, unless specified otherwise, we evaluate the performance of \sysn with 4 parallel query threads. The number of action is set to 256 ($l=8$).
We conduct all experiments on a commodity desktop computer equipped with one Core i7-4770 CPU  (4 physical cores @ 3.4 GHz, 8 MB L3 Cache shared by 8 logical cores)
and 16 GB memory (Dual channel DDR3 1600MHz).

\setlength\tabcolsep{4pt}
\begin{table*}[ht!]
    \centering
\begin{tabular}{ccc||cc|cc|cc|cc}

    \hline
    \multicolumn{3}{c||}{FIB Example} &
    \multicolumn{2}{c|}{\textbf{\sysn}} &
    \multicolumn{2}{c|}{Cuckoo} &
    \multicolumn{2}{c|}{BUFFALO} &
    \multicolumn{2}{c}{OBFs}

    \\
    Name Type  & \# Names & \# Actions &  Mem & MCQ & Mem & MCQ & Mem & MCQ & Mem & MCQ\\

    \hline
    MAC  (48 bits) & $7\! \times\! 10^5$ & 16    & 1M & 2 & 5.62M & 2     & 2.64M   & 8    & 7.36M & 15\\

    MAC  (48 bits) & $5\! \times\! 10^6$ & 256   & 16M & 2 & 40.15M & 2   & 27.70M  & 8   & 112.06M & 16\\
    MAC  (48 bits) & $3\! \times\! 10^7$ & 256   & 96M & 2 & 321.23M & 2 & 166.23M & 8 & 672.34M & 16 \\
    IPv4  (32 bits) & $1\! \times\! 10^6$  &16  & 1.5M & 2 & 4.27M & 2 & 3.77M & 8 & 10.52M & 15 \\

    IPv6  (128 bits) & $2\! \times\! 10^6$ &256 & 4M & 2 & 34.13M & 6 & 11.08M & 8 & 44.82M & 16 \\

    OpenFlow  (356b) & $3\! \times\! 10^5$ & 256 & 1M & 2 & 14.46M & 6 & 1.67M & 8 & 6.72M & 16 \\

    OpenFlow (356b) & $1.4\! \times\! 10^6$ & 65536 & 8M & 2 & 67.46M & 6 & 18.21M & 1024 & 66.60M & 17 \\

    File name (varied) & 359194 & 16 & 512K & 2 & 19.32M & 10 & 1.35M & 8 & 5.47M & 15 \\

    \hline
\end{tabular}

\caption{ Memory and query cost comparison of four FIBs and SetSep. MCQ: maximum \# of cachelines transmitted per query.
    }
    \label{tbl:memcmp}
\end{table*}

\crunchASS
\subsection{Data plane memory efficiency and MCQ}
\crunchBSS
\label{sec:dataplanememory}
Table \ref{tbl:memcmp} shows the size of memory of different types of FIBs.
For the Cuckoo hash table, we use the \texttt{(2,4)} setting.
For BUFFALO, we assume the names are evenly distributed among the actions, which gives an advantage to it.
We use the setting $k_{max}=8$. These settings are all as described or recommended in the original papers \cite{CuckooSwitch, ScaleBricks, buffalo}.

The memory space used by \sysn is significantly smaller than that of Cuckoo, BUFFALO, and OBFs. It is only determined by the number of names $n$ and the number of actions, and is independent of the name lengths.
Table \ref{tbl:memcmp} also shows the maximum number of cachelines transmitted per query (MCQ) of these FIBs.
A smaller MCQ indicates fewer data transferred from the memory to the CPU, which results in better query throughput.
\sysn always requires exactly two memory accesses per query. The other FIBs may have larger MCQ depending on the name length and number of actions.

\crunchBSS
\subsection{Memory-mode evaluation}
\crunchASS
\subsubsection{Data-plane performance}
\label{sec:memeva}

\begin{figure}[t]
    \centering
    \centering\includegraphics[width=0.70\linewidth]{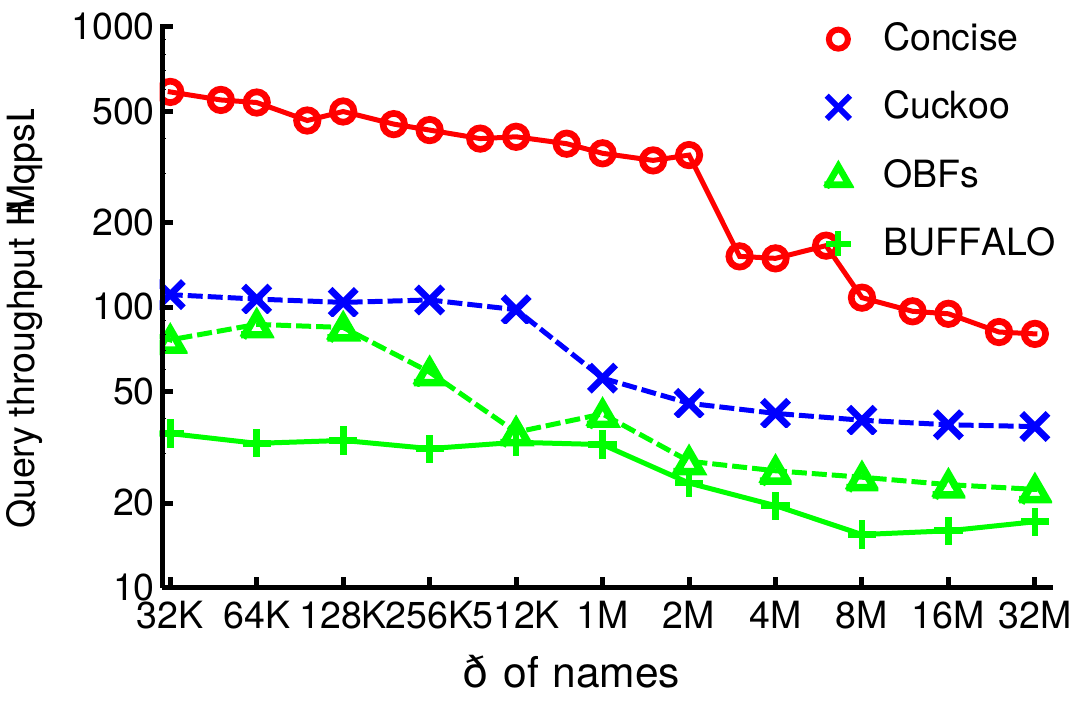}
    \caption{Query throughput versus number of names. }
    \label{fig:queryspeed1}
    \crunchAfterColFig
\end{figure}

\textbf{Query throughput versus number of names.}
Figure \ref{fig:queryspeed1} shows the query throughput of
\sysn, Cuckoo, BUFFALO, and OBFs.
The names are MAC addresses (48-bit).

    When $n$ is smaller than 2 million, the throughput of \sysn is very high ( $>$ 400M queries per second (Mqps)). This is because the memory required by \sysn is smaller than the cache size (8M for our machine). When $n\!\geq$ 2M, the throughput decreases but remains around 100 Mqps. This indicates that if other resources (e.g., I/O and buffer) are not the bottleneck, \sysn reaches 100Mqps.
The query performance decreases as the size of the query structure exceeds the CPU cache size. We observe similar results when running the evaluation on other machines with different CPUs.
Cuckoo has  the highest throughput among the remaining three FIBs but is only about only 20\% to 50\% of \sysn.
The results of Cuckoo are consistent with those presented by the original CuckooSwitch paper\footnote{The paper \cite{CuckooSwitch}\,showed a throughput 4.2x as high as our Cuckoo results on a high-end machine with two Xeon\,E5-2680\,CPUs (16 cores and 40MB L3 cache). It is approximately 4x as powerful as the one used in our experiments.}. Note that the measured time overhead includes that of query generation.\footnote{In the evaluation of 1M names,  each query of \sysn takes about 4.5 ns while generating a query takes 4.1 ns.  }

\textbf{Cost of detecting invalid names}
We also measure the cost of two approaches to detect invalid names. \ref{fig:queryspeed1} shows that using a 8-bit checksum (marked as Concise+Chk in the figure) has a minor impact on the query performance.
We provide more analysis on the approaches in Sec.\,\ref{sec:filter}.

\begin{figure}[t]
     \centering\includegraphics[width=\widthinTriColumn]{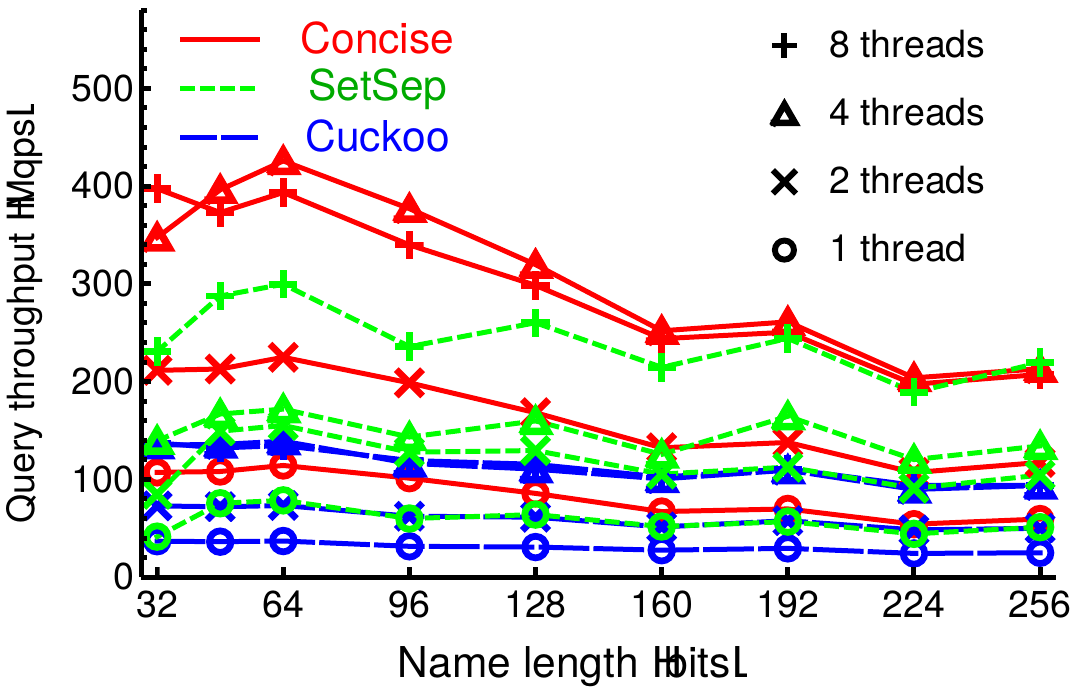}
     \caption{ Query throughput vs. name length}
    \label{fig:queryspeedKeylength}
\end{figure}

\textbf{Query throughput versus name length and number of CPU cores.}
Figure \ref{fig:queryspeedKeylength} shows the query throughput using different name lengths. Each FIB contains 256K names.
As the length grows, the throughput of all types of \sysn and Cuckoo FIBs decreases. Note that the memory size of \sysn is independent of the name length. Hence, the throughput decrease of \sysn is due to the increase of hashing time.
One interesting observation is that when the length is a multiple of 64 bits, the query throughput of \sysn is slightly increased.
This is mainly because the experiments are conducted on a 64-bit CPU.
The query throughput grows approximately linearly to the number of used threads, as long as the number of threads does not exceed the number of physical CPU cores of the platform.

\begin{figure}[t]
    \centering\includegraphics[width=\widthinTriColumn]{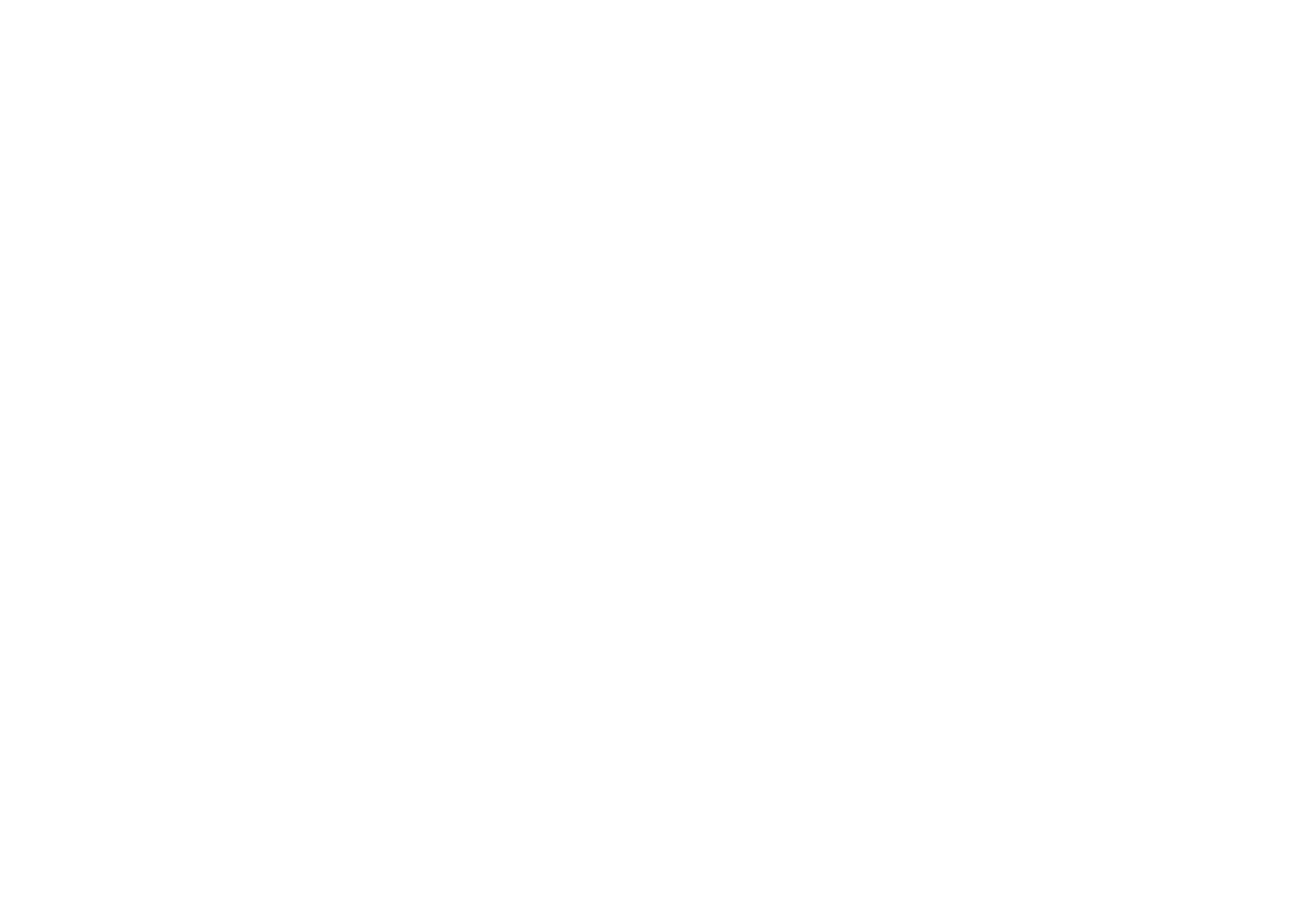}
     \crunchBeforeTriFigCap
    \caption{ \sysn query throughput under different update rates}
    \label{fig:updateThrComb}
\end{figure}

\textbf{Query throughput during updates.}
Figure \ref{fig:updateThrComb} shows the throughput of \sysn during updates, including name additions, deletions, and action changes.
There is only very small decrease of query throughput even when the update frequency is as high as hundreds of thousands of names updated per second.
    We mark the one-$\sigma$ (68\%) confidence interval of the throughput when there is no concurrent query in Figure \ref{fig:updateThrComb}. Evaluation result shows that the throughput of \sysn still remains in its normal range during updates.
For \sysn with 4M names the throughput downgrade is negligible.

\begin{figure}[t]
    \centering\includegraphics[width=\widthinTriColumn]{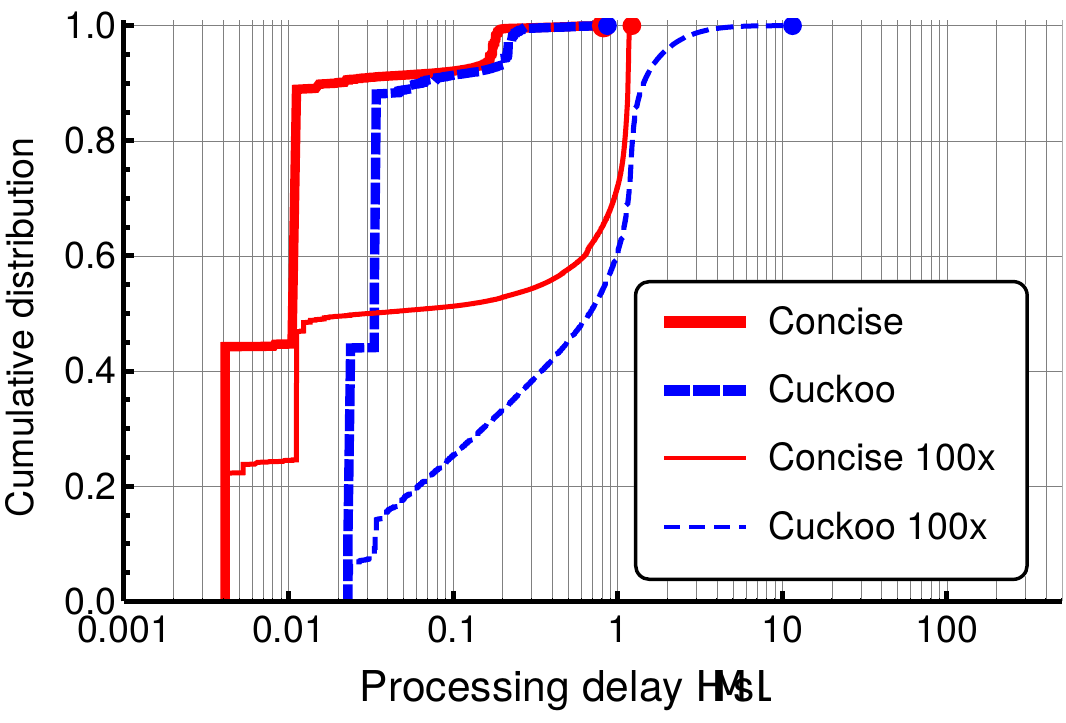}
     \crunchBeforeTriFigCap
    \caption{ CDF of the processing delay of \sysn and Cuckoo}
    \label{fig:CDFLatencyComb}
\end{figure}

\textbf{Processing delay.}
 We conduct event-based simulations of packet processing on the data plane to study the process delay.
We simulate a single-thread processor with two-level cache mechanism. The
packets are processed in a first-come, first-served fashion. Each packet consists of the header and payload.
The packets are put in a queue upon reception and wait to be processed by the prosessor.
We measure the processing delay for real traffic data from the CAIDA\hyphenation{Anony-mized} Anonymized Internet Traces of December 2013~\cite{CAIDAdata}.
  The average packet rate is about 210K packets per second.
  In Figure \ref{fig:CDFLatencyComb}, \sysn has smaller processing delay than Cuckoo before the 90th percentile, but they have similar tails.
  To study the processing delay under larger traffic volumes, we replay the trace 100x as fast as the original. Shown as the thin curves,
the processing delay of \sysn is clearly smaller than that of Cuckoo before the 60th percentile. After that, the two curves are similar, except that Cuckoo has a longer tail. Overall, the processing delay of \sysn is very small ($< 1\mu$s) even under high data volumes.

\begin{figure}[t]
    \centering\includegraphics[width=\widthinTriColumn]{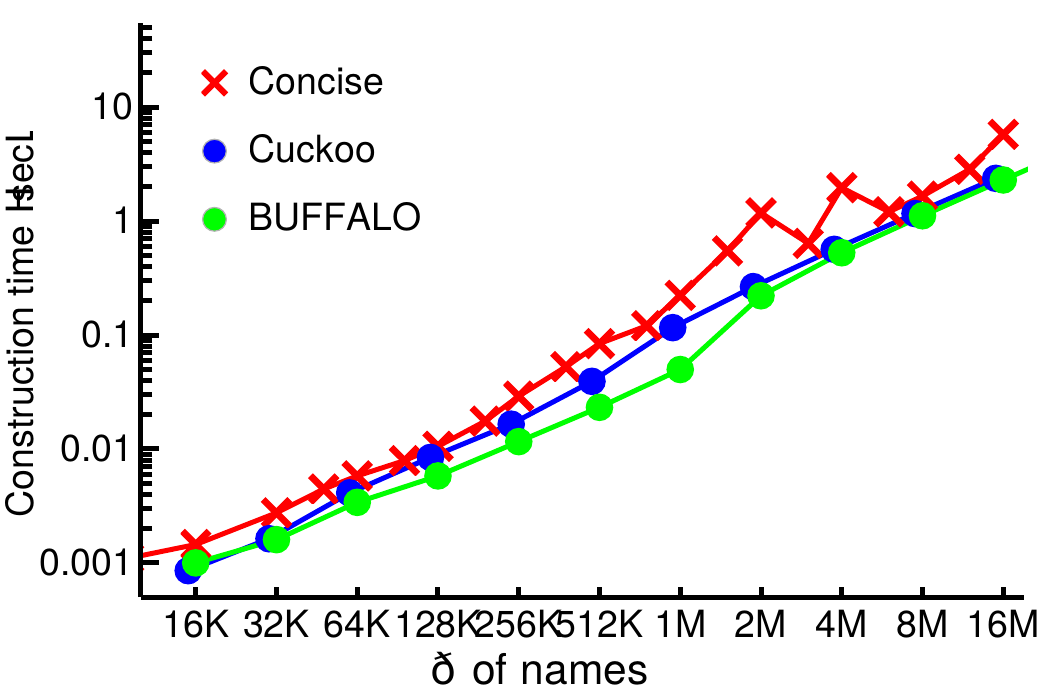}
    \caption{ Construction time comparison among three FIBs}
    \label{fig:buildtimecomp}
\end{figure}

\subsubsection{Control plane performance}
\textbf{Construction time.}
Figure \ref{fig:buildtimecomp} shows the average time to construct the query and control structures for one switch with various number of names.
The construction time of \sysn grows approximately linearly to the number of addresses.
Although the time of \sysn is larger than that of Cuckoo and BUFFALO, it is still very small.
For 4M names, it takes only 1 second to construct the FIB.
Note that the graph $G$ can be reused for all other switches in the network. Hence, network-wide FIB construction only takes few seconds.

\begin{figure}[t]
    \centering\includegraphics[width=\widthinTriColumn]{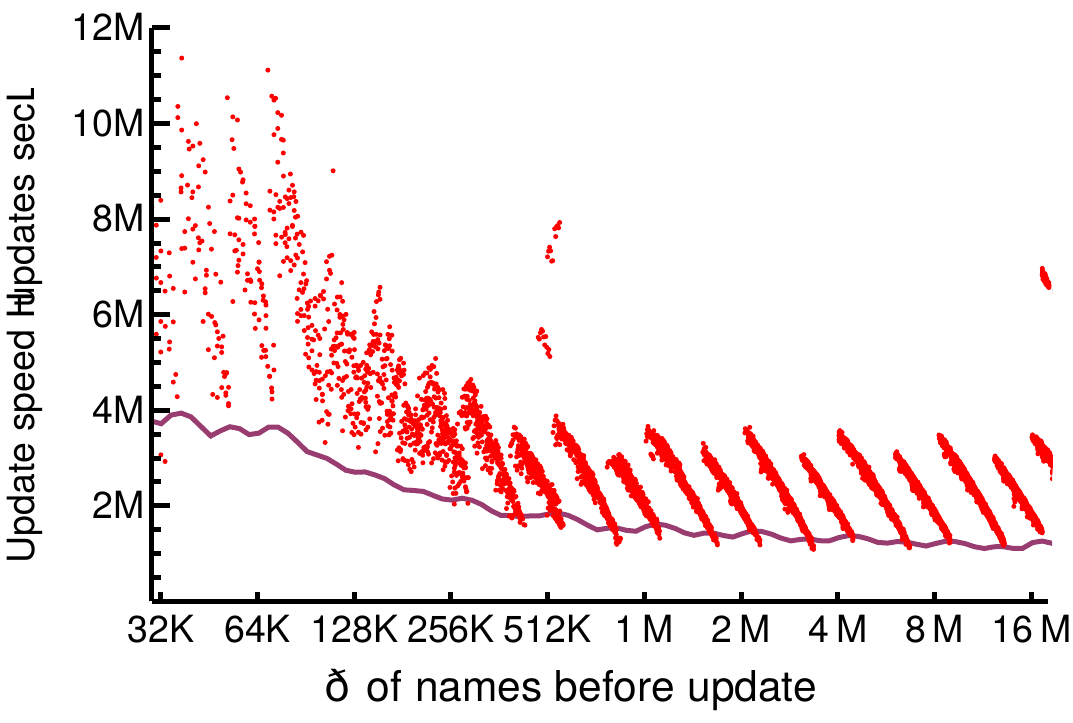}
\caption{ Update speed. Line: avg. spd. including POG reconstruction.}
    \label{fig:updateInsert}
    \end{figure}

\textbf{Update speed.} The update speed indicates the ability to react to network dynamics.
All types of network dynamics, including host and link changes, are reflected as name additions, deletions, and action changes in the FIBs.
Figure \ref{fig:updateInsert} shows the update speed of \sysn in number of updates processed per second.
We vary the number of names before update and measure the time used to insert a number of new names.
Each run of the experiment is shown as a point in the figure.
In most cases, it reaches at least  1M updates per second, which is sufficient for very large networks.

\textbf{On POG reconstruction.}
In some rare cases, adding a new name may require reconstruction of the POG when it introduces a new cycle in to the bipartite graph. This may take non-negligible time (0.2 seconds when there are 1M names). Theoretical results show this happens with probability less than $\frac{1.5}{n}$. This value is even smaller in practice (about 1.3 parts per million when there are 1M names). Note that, POG reconstruction may happen only when there is a new name added to the network. \textit{Modifying a forwarding action of a existing name (or removing a name) never results in POG reconstruction.}
The line in Fig. \ref{fig:updateInsert} shows the average update speed (including the time overhead for reconstruction).
POG reconstruction only imposes minor impact on the update speed.

\textbf{Network-wide shared bipartite graph.}
    For some networks that require every switch to store all destination names such as Ethernet, the name set $S$ is identical for all switches in the network.
    Hence, all switches in the network may share the same $G$ and $\langle h_a, h_b \rangle$.
    Constructing and updating the FIBs in all switches only require computing $G$ once.
    e.g., the phase I of the $\mathtt{construct}$ procedure (Sec. \ref{sec:construct}) is only executed \textit{once} for FIBs of all switches in the network. This indicates that the construction time overhead for FIBs of multiple switches can be further reduced.
    Note that for a single switch, the time used for phase I is about half of the total of $\mathtt{construct}$.

\textbf{Communication overhead. }
We compute the entropy of the information included in update messages in Table \ref{tab:messagelength}.
The update message length grows logarithmically with respect to either the number of names $n$ or the number of actions.
The communication overhead of \sysn
is smaller than that of most OpenFlow operations.
\begin{table}[h]
    \centering

    \begin{tabular}{ccc}

        & $n\!=\!3\!\times\!10^5$  & $n\!=\!1.4\! \times\! 10^6$ \\
        &   $2^8$ actions & $2^{16}$ actions \\

\hline
{Name addition}& 75.2 & 107.2 \\
Action change & 65.6 & 88.8 \\
\hline
\end{tabular}
\caption{Entropy of one update message in bits}
\label{tab:messagelength}
\end{table}

\subsection{Prototype Implementation and Evaluation}

    \begin{figure}[t]
    \centering
    \includegraphics[width=\widthinTriColumn]{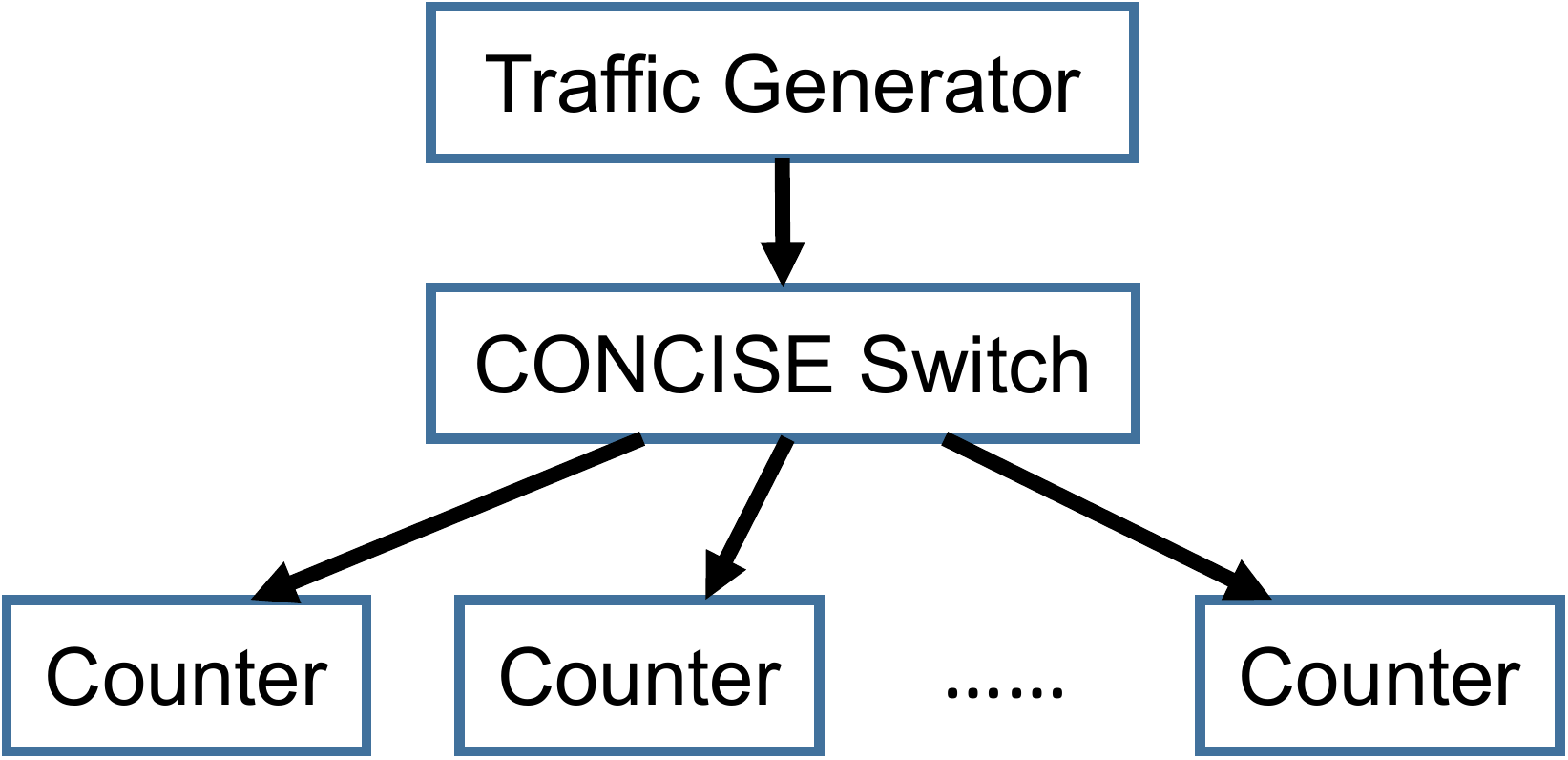}
    \crunchBeforeTriFigCap
    \caption{ \sysn prototype with Click modular router}
    \label{fig:Clicksys}
    \end{figure}

    \begin{figure}[t]
    \centering
    \includegraphics[width=\widthinTriColumn]{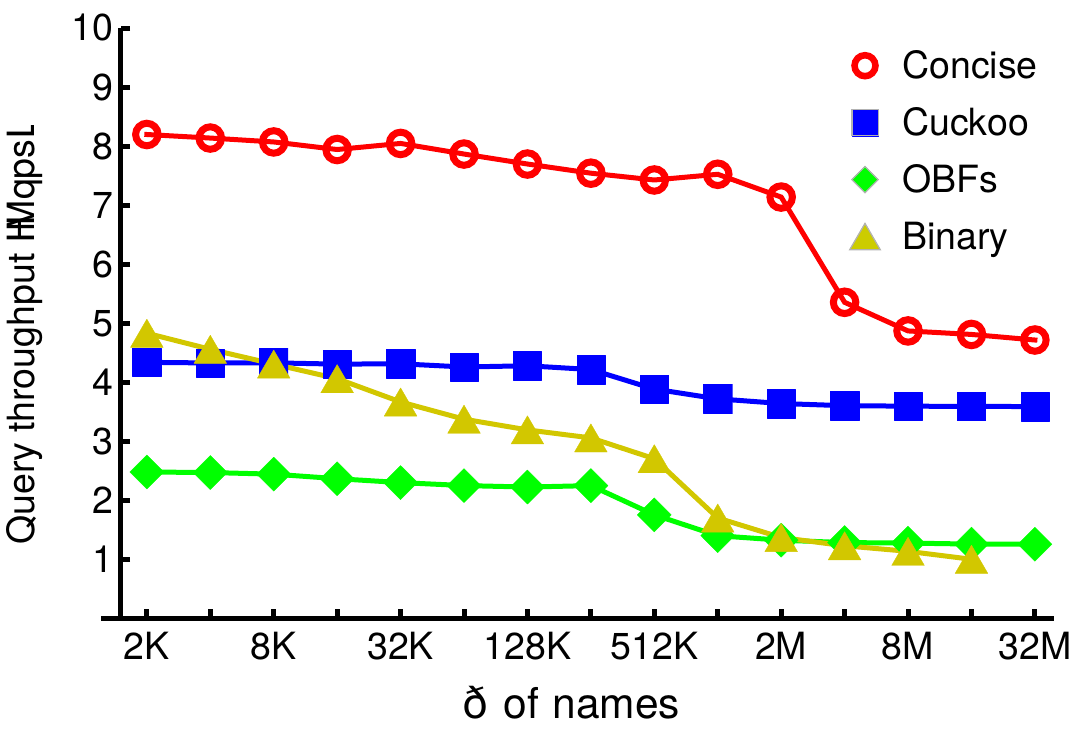}
    \crunchBeforeTriFigCap
    \caption{ Forwarding throughput comparison on Click}
    \label{fig:Click1}
    \end{figure}
\subsubsection{Implementation on Click}

We implement a \sysn prototype on Click Modular Router \cite{click}. The structure of the prototype system is as shown in Figure \ref{fig:Clicksys}.
It receives packets from one inbound port and forwards each packet to one of its several outbound ports.
Upon receiving a packet, it queries the POG using the address field of the packet, i.e., the name, and decides the outbound port of the packet.
In addition, we implement the \texttt{(2,4)}-Cuckoo hash table, OBFs, as well as the binary search mechanism on Click.
Figure \ref{fig:Click1} shows the forwarding throughput.
The Click modules in each evaluation includes one traffic generator generating packets with valid 64-bit names, one switch that executes queries on the FIB, and packet counters connected to the egress ports of the switch.
The experiments are conducted on one CPU core.

Results show that  \sysn  always has the highest throughput. When $n<2\text{M}$, \sysn is smaller than the cache size and the query throughput is about 2x as fast as Cuckoo and 4x as fast as OBFs. When $n\geq2\text{M}$, the throughput of \sysn is still the highest.
 Meanwhile, \sysn  uses much less  memory, about 10\% to 20\% of that of  Cuckoo, OBFs, and Binary.

\begin{figure}[t]
\centering
    \includegraphics[width=0.8\linewidth]{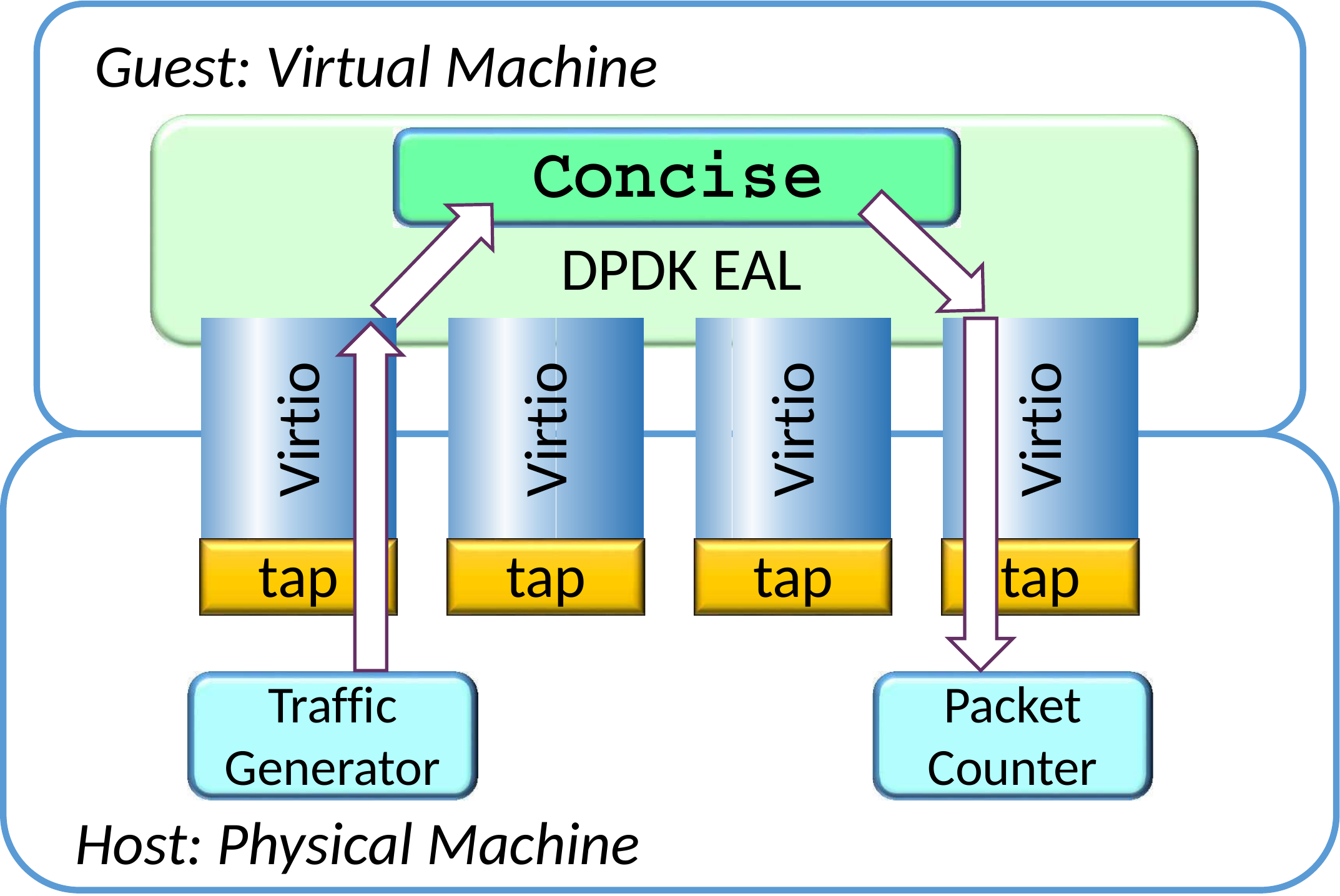}
    \crunchBeforeTriFigCap
    \caption{ \sysn prototype on DPDK}
  \label{fig:dpdksystem}
\end{figure}
\begin{figure}[t]
\centering\includegraphics[width=\widthinTriColumn]{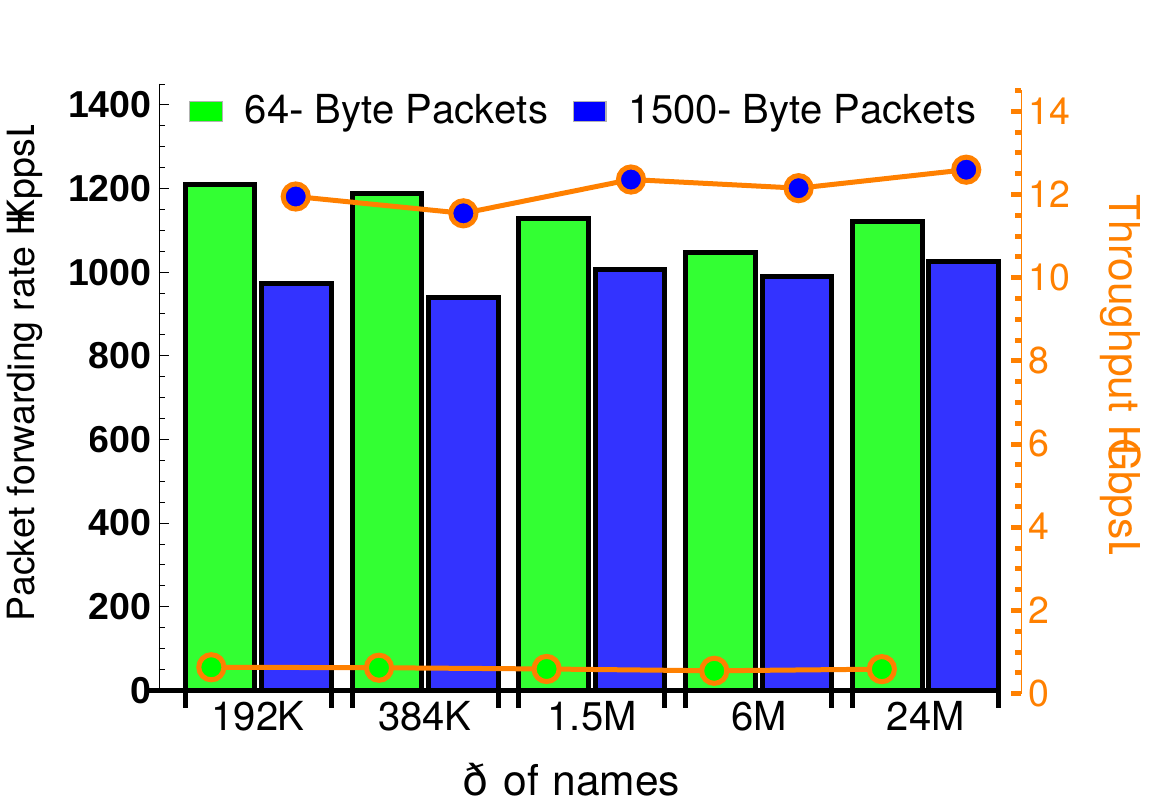}
    \centering
    \caption{ Performance of the \sysn prototype on DPDK }
    \label{fig:dpdkthrput}
    \end{figure}
    \begin{figure}[t]
    \centering
    \includegraphics[width=\widthinTriColumn]{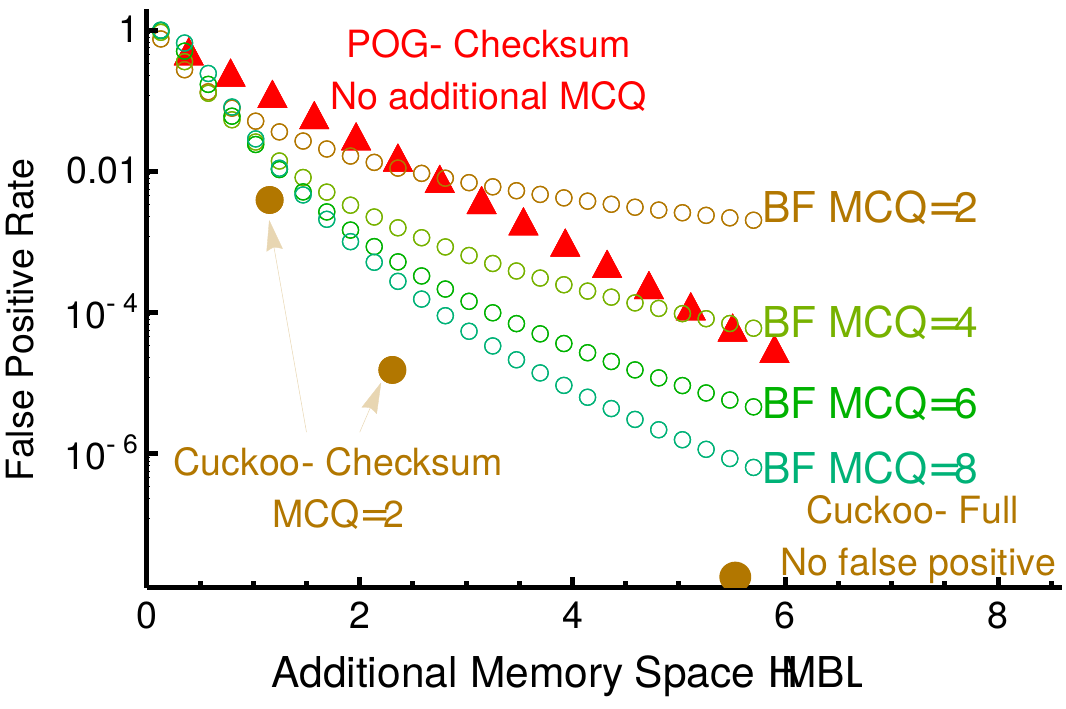}
    \crunchBeforeTriFigCap
\caption{ Approaches of detecting invalid names}
\label{fig:falsepositive}
\vspace{-2ex}
\end{figure}

\subsubsection{Implementation with DPDK}
We also build a \sysn prototype on the hardware Environment Abstraction Layer (EAL) provided by DPDK.
It maintains a POG query structure. The query structure is initialized during boot up and can be updated upon network dynamics.
The prototype reads packets from the inbound ports, executes queries on the query structure, and then forwards each packet to the corresponding outbound port.

We implement both the traffic generator and FIB application on the same commodity computer using virtualization techniques. As shown in Figure \ref{fig:dpdksystem}, we create a guest virtual machine (VM) on the host machine using KVM and Qemu to install \sysn. The VM is equipped with four virtio-based virtual network interface cards. Linux TAP kernel virtual devices are attached to the virtio devices on the host side. The programs running on the host machine communicate with the guest VM via the Linux TAPs.
On the host machine, we use a traffic generator program to send raw Ethernet packets to \sysn running on the VM. The host machine receives the forwarded packets from \sysn and counts the number of packets using default counters provided by the Linux system.

We measure the throughput of \sysn with different numbers of names.
The barchart in Figure \ref{fig:dpdkthrput} shows that \sysn is able to generate, forward, and receive more than 1M packets per second, for both 64-Byte and 1500-Byte packets.
The forward throughput is at least 12 Gbps for  1500-Byte Ethernet packets.
The throughput of Cuckoo is only 60\% to 80\% of the throughput of \sysn.
The forwarding throughput does not significantly change when the number of names grows or packet length changes.
This indicates that the impact of \sysn on the overall performance is so small that it is negligible compared to the other overheads. The bottleneck of this evaluation is on other parts of processing, e.g., data transmission between the host  machine and guest VM.
We expect a much higher throughput on physical NICs.
\section{Discussion}
\label{sec:dis}
\crunchBSS
\subsection{Deal with Alien Names}
\crunchASS
\label{sec:filter}
An alien name is a name that is not in $S$ during \sysn construction.
Querying an alien name may result in an arbitrary forwarding action. Compared to the forwarding table miss of Ethernet, which let the packets flood to all interfaces,\,\sysn causes no flooding.
Operators may choose one or some of the following mechanisms to detect the alien names.
\begin{itemize}
        \crunchitemize
\item At an ingress switch, every incoming packet should be checked by a filter or firewall to validate that its destination does exist in the network.  This filter can be implemented as a network function running on the border of the network, and can be integrated with the firewall. 
        \crunchitemize
\item  Maintain a Bloom filter at each of the switches. Packets with valid names pass this filter and are then processed by \sysn FIB. 
        \crunchitemize
    \item
        In addition to the $l$-bit query results, also maintain the checksums for each name in the \sysn FIBs.
        Adding checksums will increase the memory size of \sysn. For $r$-bit checksums, the overall memory cost of a query structure is $2(l+r)m+O(1)$. Note that as long as $l+r$ does not exceed the word length of the computing platform, the time overhead of all operations remains unchanged.
\crunchitemize
\end{itemize}

Assuming there are in total 1M names. Fig\,\ref{fig:falsepositive} compares the memory and computational overheads of the above approaches. 
The false positive rate can be controlled to be as low as $10^{-5}$ with $<2$MB memory overhead using the filter of Cuckoo with checksums.
The performance when using Bloom filters may vary depending on the parameters.
We also recommend to utilize the time-to-live (TTL) value of to prevent the packet being forwarded in the network forever.

The unique property of returning an arbitrary value for an alien name may also be useful for \sysn as a network load balancer: for a server-visiting flow that is new to the network, \sysn can forward it to one of the servers with adjustable weights.

\crunchBSS
\subsection{\sysnintitlelarge ~versus Cuckoo and SetSep}
\crunchASS
\label{sec:comparison}

\sysn is essentially a classifier for names, and each class represents a forwarding action. \sysn does not store the names.
 Cuckoo stores all names and actions in a key-value store.
 
  SetSep has some  properties  similar to \sysn.
Both of them do not store names and return meaningless results for unknown names.
In ScaleBricks \mbox{\cite{ScaleBricks}}, SetSep is only used as a separator to distribute the FIB to different computers, rather than the FIB. Meanwhile, the update scheme for SetSep is not explicitly explained \mbox{\cite{SetSepHotOS,ScaleBricks}}, and there is no discussion about handling dynamic FIB size growth.

In addition to the memory size results in Table 1, we show some comparison results of SetSep in what follows.
The construction speed of SetSep is slower than that of \sysn and Cuckoo by more than an order of magnitude:
10 seconds for one single FIB of 1M names in our experiments.
We also measure the update speed of SetSep without adding new names, which turns to be less than 10K/s ($<$ 1 \% of \sysn).
The query speed of SetSep is higher than that of Cuckoo. SetSep needs to compute $1+l$ hash values and read $2+2l$ values for each query.
We implement a static  SetSep with 1.4M names and $l=8$, using 2.19MB memory. Its query throughput is 211 Mqps using 4 threads. In comparison, \sysn with the same settings uses 4M memory and reaches 470 Mqps.

In addition, we summarize the reasons of the performance gain of \sysn as follows.
(1) \dtsn does \textit{not} maintain a copy of the names in the query structure. The memory size of the query structure is much smaller than the other solutions.
    \sysn demonstrates higher cache-hit rate, which leads to better performance on cache-based systems.
(2) The query procedure  does not contain any branches (e.g, $\mathtt{if}$ statements). This helps the CPU to predict and execute the instructions in the query procedure.
    (3) The efficient concurrency control mechanism further improves the query speed of \sysn.

\begin{TON}
\subsection{Example Use Case}
\sysn provides desired FIB properties for many current and future architecture designs that adopt flat names as mentioned in Sec. \ref{sec:intro}. We present a use case where it can be applied in a large enterprise network.

 A large enterprise or data center network may include up to millions of end hosts and more VMs \cite{B4}. In these networks, internal flows contribute to the most bandwidth, which can be forwarded by \sysn using destination names on Layer 2.
The destination of a packet  in this network can only be either a host or a gateway.
We require  hosts in the network voluntarily check the validity of the packets before sending them out. This can be easily achieved using software firewalls such as \textit{iptables}.

As of the gateway, we require it to execute two network functions: (1) For packets going out from the network, perform Layer 3 routing using the external IP of the destination. This is a basic function a router. (2) For packets going into the network, filter out all packets with invalid destinations. This can be implemented by a firewall. The packets will be forwarded using the Layer 2 names of the destinations. 
In addition, we require all packets in the network to carry a time-to-live (TTL) value to prevent packets from being forwarded forever in case packets with invalid names pass the firewalls.

\end{TON}

\section{Conclusion}
\label{sec:conclusion}
\sysn is a portable FIB design for network name lookups, which is developed based on a new algorithm Othello Hashing.
\sysn minimizes the memory cost of FIBs and moves the construction and update functionalities to the SDN controller.
We implement \sysn using three platforms.
According to our analysis and evaluation, \sysn uses the smallest memory to achieve the fastest query speed among existing FIB solutions for name lookups.
As a fundamental network algorithm, we expect that Othello Hashing will be used in a large number of network systems and applications where existing tools such as Bloom Filters and Cuckoo Hashing may not be suitable.


\bibliographystyle{abbrv}
\bibliography{paperof}

\begin{thebibliography}{10}

\bibitem{CAIDAdata}
{The CAIDA UCSD Anonymized Internet Traces}.
\newblock http://www.caida.org/data/passive/passive{\_}2013{\_}dataset.xml.

\bibitem{Anand2010}
A.~Anand, C.~Muthukrishnan, S.~Kappes, A.~Akella, and S.~Nath.
\newblock {Cheap and large CAMs for high performance data-intensive networked
  systems}.
\newblock In {\em Proc. of USENIX NSDI}, 2010.

\bibitem{AIP}
D.~G. Anderson, H.~Balakrishnan, N.~Feamster, T.~Koponen, D.~Moon, and
  S.~Shenker.
\newblock {Accountable Internet Protocol (AIP)}.
\newblock In {\em Proc. of ACM SIGCOMM}, 2008.

\bibitem{Asai2015}
H.~Asai and Y.~Ohara.
\newblock {Poptrie: A Compressed Trie with Population Count for Fast and
  Scalable Software IP Routing Table Lookup}.
\newblock In {\em Proc. of ACM SIGCOMM}. ACM, 2015.

\bibitem{layered}
H.~Balakrishnan, K.~Lakshminarayanan, S.~Ratnasamy, S.~Shenker, I.~Stoica, and
  M.~Walfish.
\newblock {A layered naming architecture for the Internet}.
\newblock In {\em Proc. of ACM SIGCOMM}, 2004.

\bibitem{MonotoneHash}
D.~Belazzougui, P.~Boldi, R.~Pagh, and S.~Vigna.
\newblock {Monotone minimal perfect hashing: searching a sorted table with O
  (1) accesses}.
\newblock In {\em Proc. of ACM SODA}. Society for Industrial and Applied
  Mathematics, 2009.

\bibitem{BloomierFilter}
O.~B. {Bernard Chazelle, Joe Kilian, Ronitt Rubinfeld, Ayellet Tal}.
\newblock {The Bloomier Filter: An Efficient Data Structure for Static Support
  Lookup Tables}, 2004.

\bibitem{Botelho2012}
F.~C. Botelho, N.~Wormald, and N.~Ziviani.
\newblock {Cores of random r-partite hypergraphs}.
\newblock {\em Inf. Process. Lett.}, 112(8-9):314--319, apr 2012.

\bibitem{ROFL}
M.~Caesar, T.~Condie, J.~Kannan, K.~Lakshminarayanan, I.~Stoica, and
  S.~Shenker.
\newblock {ROFL: Routing on Flat Labels}.
\newblock In {\em Proc. of ACM SIGCOMM}, 2006.

\bibitem{payless}
S.~R. Chowdhury, M.~F. Bari, R.~Ahmed, and R.~Boutaba.
\newblock {PayLess: A Low Cost Netowrk Monitoring Framework for Software
  Defined Networks}.
\newblock In {\em Proc. of IEEE/IFIP NOMS}, 2014.

\bibitem{devroye2003cuckoo}
L.~Devroye and P.~Morin.
\newblock Cuckoo hashing: further analysis.
\newblock {\em Information Processing Letters}, 86(4):215--219, 2003.

\bibitem{SetSepHotOS}
B.~Fan, D.~Zhou, H.~Lim, M.~Kaminsky, and D.~G. Andersen.
\newblock {When cycles are cheap, some tables can be huge}.
\newblock In {\em Proc. of USENIX HotOS}, 2013.

\bibitem{Greenberg2009}
B.~A. Greenberg et~al.
\newblock {VL2: a scalable and flexible data center network}.
\newblock {\em ACM SIGCOMM CCR}, 09:51--62, 2009.

\bibitem{DPDK}
Intel.
\newblock {Data Plane Development Kit}.
\newblock http://dpdk.org/.

\bibitem{CCN}
V.~Jacobson, D.~Smetters, J.~D. Thornton, M.~F. Plass, N.~H. Briggs, and R.~L.
  Braynard.
\newblock {Networking Named Content}.
\newblock In {\em Proc. of ACM CoNEXT}, 2009.

\bibitem{B4}
S.~Jain and Others.
\newblock {B4: Experience with a Globally-Deployed Software Defined WAN}.
\newblock In {\em Proc. of ACM SIGCOMM}, 2013.

\bibitem{janson2008susceptibility}
S.~Janson and M.~J. Luczak.
\newblock {Susceptibility in subcritical random graphs}.
\newblock {\em J. Math. Phys.}, 49(12):125207, 2008.

\bibitem{headerspace1}
P.~Kazemian, G.~Varghese, and N.~McKeown.
\newblock {Header Space Analysis: Static Checking For Networks}.
\newblock In {\em Proc. of USENIX NSDI}, 2012.

\bibitem{Seattle}
C.~Kim, M.~Caesar, and J.~Rexford.
\newblock {Floodless in seattle: a scalable ethernet architecture for large
  enterprises}.
\newblock In {\em Proc. of SIGCOMM}, 2008.

\bibitem{click}
E.~Kohler, R.~Morris, and B.~Chen.
\newblock {\em {The Click Modular Router}}.
\newblock PhD thesis, Massachusetts Institute of Technology, 2000.

\bibitem{MWHC1996}
B.~S. Majewski, N.~Wormald, G.~Havas, and Z.~Czech.
\newblock {A Family of Perfect Hashing Methods}.
\newblock {\em Comput. J.}, jun 1996.

\bibitem{Caesar}
M.~Moradi, F.~Qian, Q.~Xu, Z.~M. Mao, D.~Bethea, and M.~K. Reiter.
\newblock {Caesar: High-Speed and Memory-Efficient Forwarding Engine for Future
  Internet Architecture}.
\newblock In {\em Proc. of ACM/IEEE ANCS}, 2015.

\bibitem{HIP}
R.~Moskowitz, P.~Nikander, P.~Jokela, and T.~Henderson.
\newblock {Host Identity Protocol}.
\newblock Technical report, 2008.

\bibitem{CuckooHashing}
R.~Pagh and F.~F. Rodler.
\newblock {Cuckoo hashing}.
\newblock {\em J. Algorithms}, 51(2):122--144, may 2004.

\bibitem{ROME}
C.~Qian and S.~Lam.
\newblock {ROME: Routing On Metropolitan-scale Ethernet}.
\newblock In {\em Proc. of IEEE ICNP}, 2012.

\bibitem{MobilityFirst}
D.~Raychaudhuri, K.~Nagaraja, and A.~Venkataramani.
\newblock {MobilityFirst: A Robust and Trustworthy Mobility Centric
  Architecture for the Future Internet}.
\newblock {\em MC2R}, 2012.

\bibitem{Saltzer}
J.~Saltzer.
\newblock {On the naming and binding of network destinations}.
\newblock RFC 1498, 1993.

\bibitem{Shahbaz2016}
M.~Shahbaz, S.~Choi, B.~Pfaff, C.~Kim, N.~Feamster, N.~McKeown, and J.~Rexford.
\newblock {PISCES: A Programmable, Protocol-Independent Software Switch}.
\newblock In {\em Proc. of ACM SIGCOMM}, 2016.

\bibitem{Disco}
A.~Singla, P.~B. Godfrey, K.~Fall, G.~Iannaccone, and S.~Ratnasamy.
\newblock {Scalable Routing on Flat Names}.
\newblock In {\em Proc. of ACM CoNEXT}, 2010.

\bibitem{Srinivasan1999}
V.~Srinivasan, S.~Suri, and G.~Varghese.
\newblock {Packet classification using tuple space search}.
\newblock In {\em Proc. of ACM SIGCOMM}, 1999.

\bibitem{PAST}
B.~Stephens, A.~Cox, W.~Felter, C.~Dixon, and J.~Carter.
\newblock {PAST: Scalable Ethernet for Data Centers}.
\newblock In {\em Proc. of ACM CoNEXT}, 2012.

\bibitem{Wang2013}
Y.~Wang et~al.
\newblock {Wire speed name lookup: a GPU-based approach}.
\newblock {\em Proc. of USENIX NSDI}, 2013.

\bibitem{SAIL}
T.~Yang, G.~Xie, Y.~Li, Q.~Fu, A.~X. Liu, Q.~Li, and L.~Mathy.
\newblock {Guarantee IP Lookup Performance with FIB Explosion}.
\newblock In {\em Proc. of ACM SIGCOMM}, 2014.

\bibitem{buffalo}
M.~Yu, A.~Fabrikant, and J.~Rexford.
\newblock {BUFFALO: Bloom filter forwarding architecture for large
  organizations}.
\newblock In {\em Proc. of ACM CoNEXT}, 2009.

\bibitem{DIFANE}
M.~Yu, J.~Rexford, M.~J. Freedman, and J.~Wang.
\newblock {Scalable flow-based networking with DIFANE}.
\newblock In {\em Proc. of ACM SIGCOMM}, 2010.

\bibitem{ConciseICNP}
Y.~Yu, D.~Belazzougui, C.~Qian, and Q.~Zhang.
\newblock In {\em Proc. of IEEE ICNP}, 2017.

\bibitem{NDN}
L.~Zhang, D.~Estrin, J.~Burke, V.~Jacobson, J.~D. Thornton, D.~K. Smetters,
  B.~Zhang, G.~Tsudik, D.~Massey, C.~Papadopoulos, T.~Abdelzaher, L.~Wang,
  P.~Crowley, and E.~Yeh.
\newblock {Named data networking (ndn) project}.
\newblock {\em NDN Tech. Rep.}, 2010.

\bibitem{ScaleBricks}
D.~Zhou, B.~Fan, H.~Lim, D.~G. Anderson, M.~Kaminsky, M.~Mitzenmacher, R.~Wang,
  and A.~Singh.
\newblock {Scaling Up Clustered Network Appliances with ScaleBricks}.
\newblock In {\em Proc. of ACM SIGCOMM}, 2015.

\bibitem{CuckooSwitch}
D.~Zhou, B.~Fan, H.~Lim, M.~Kaminsky, and D.~G. Anderson.
\newblock {Scalable, High Performance Ethernet Forwarding with CuckooSwitch}.
\newblock In {\em Proc. of ACM CoNEXT}, 2013.

\end{thebibliography}
\normalsize


\begin{IEEEbiography}[{\includegraphics[width=1in,height=1.25in,clip,keepaspectratio]{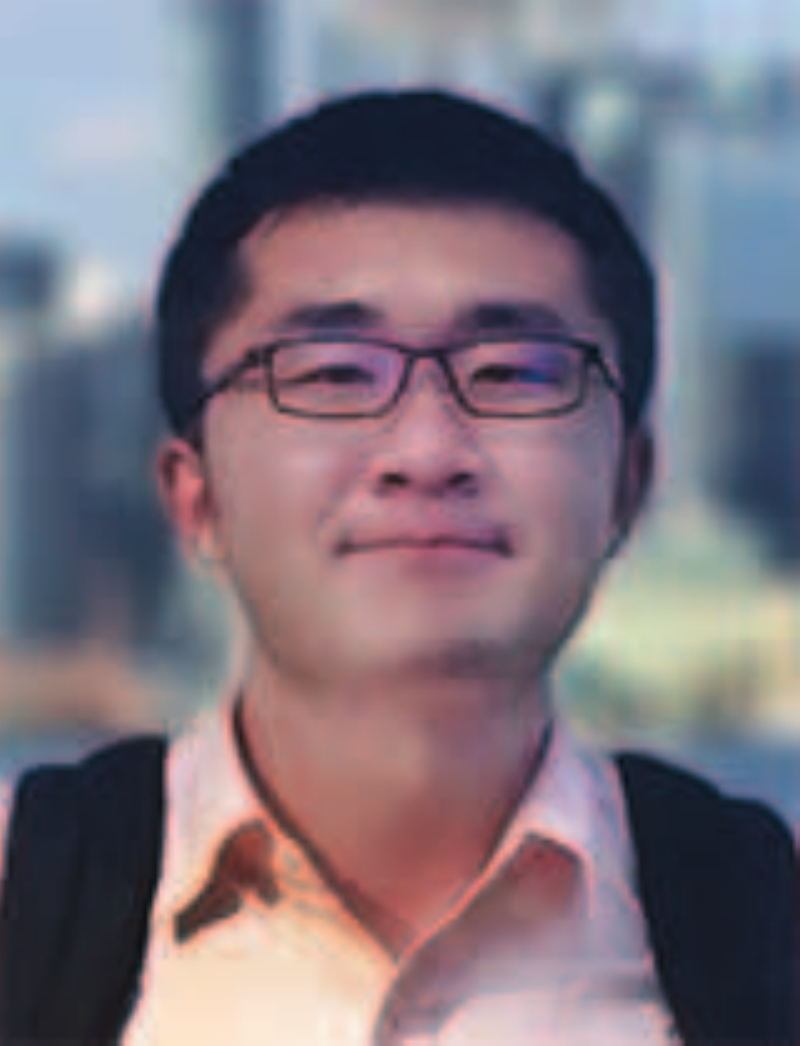}}]{Ye Yu}(S'13)
is a Ph.D. student at the Department of Computer Science, University of Kentucky. He received the B.Sc. degree  from Beihang University. His research interests including data center networks, software defined networking.  Especially, he is doing research about applications of fast and memory-effective hashing applications in computer networking systems and data storage.
\end{IEEEbiography}
\vspace{-5ex}
\begin{IEEEbiography}[{\includegraphics[width=1in,height=1.25in,clip,keepaspectratio]{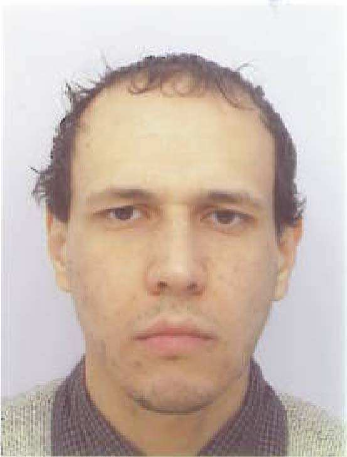}}]{Djamal}  Djamal Belazzougui is currently a researcher at CERIST research centre, Algeria.
He received an enginerring degree from the national high school
of Computer science, Algeria, and earned a Phd degree from Paris-VII,
Paris-Diderot university, France.
He subsequently spent three years as a postdoctoral researcher at the
University of Helsinki, Finland.
His research topics include hashing, succinct and compressed data
structures and string algorithms.
\end{IEEEbiography}
\vspace{-5ex}
\begin{IEEEbiography}[{\includegraphics[width=1in,height=1.25in,clip,keepaspectratio]{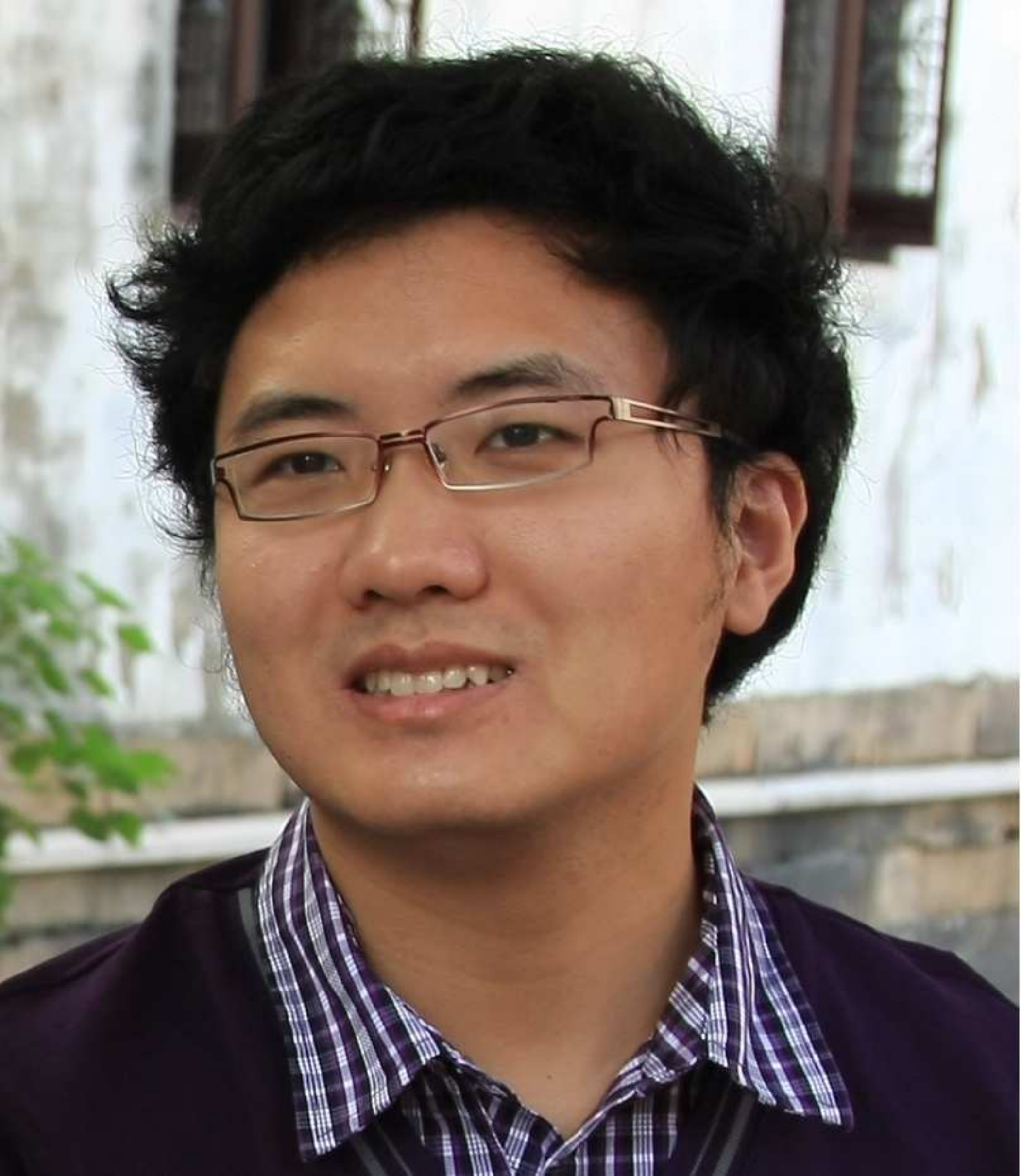}}]{Chen Qian} (M'08) is an Assistant Professor
at the Department of Computer Engineering, University of California Santa Cruz. He received the B.Sc. degree  from
Nanjing University in 2006, the M.Phil. degree  from the Hong Kong University of Science and Technology in 2008, and the Ph.D. degree
from the University of Texas at Austin in 2013, all in Computer Science. His research interests include computer networking, network security, and Internet of Things. He has published more than 60 research papers in highly competitive conferences and journals. He is a member of IEEE and ACM.
\end{IEEEbiography}
\vspace{-5ex}
\begin{IEEEbiography}[{\includegraphics[width=1in,height=1.25in,clip,keepaspectratio]{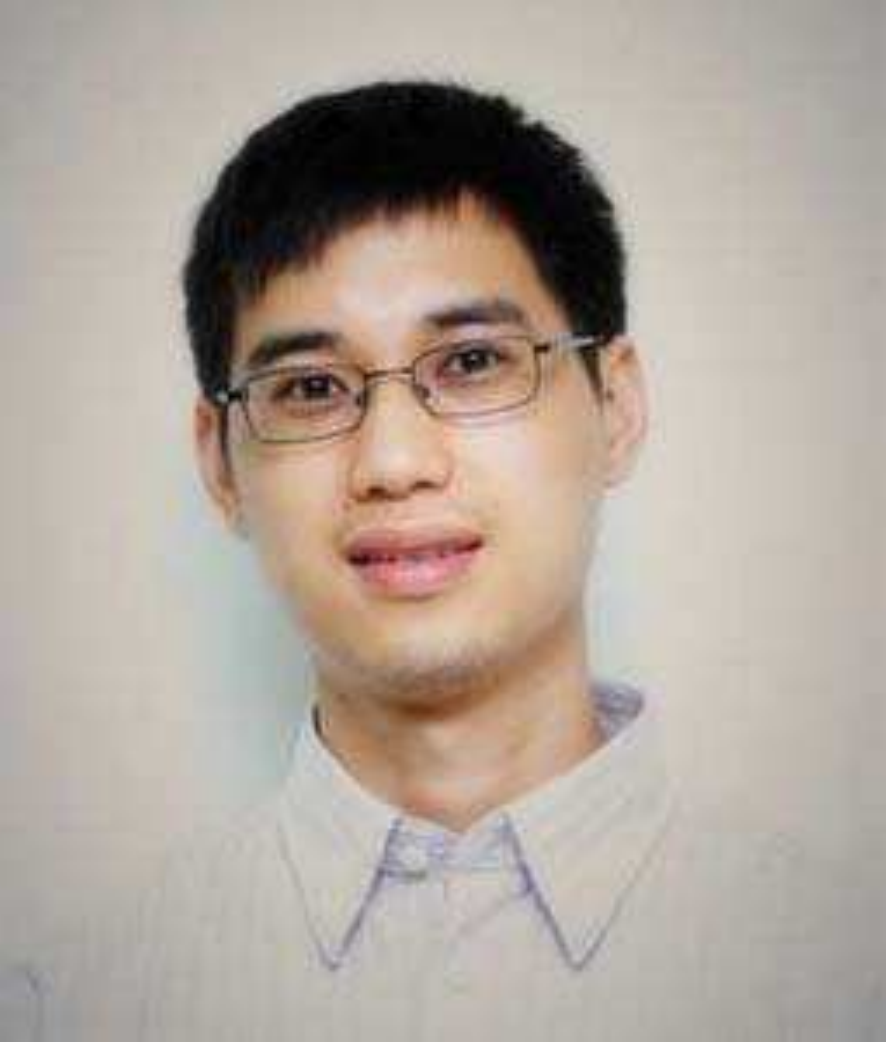}}]{Qin Zhang}  Qin Zhang is currently an Assistant Professor at Indiana University Bloomington.  He received the B.S. degree from Fundan University and the Ph.D. degre from Hong Kong University of Science and Technology.  He also spent a couple of years as a post-doc at Theory Group, IBM Almaden Research Center, and Center for Massive Data Algorithmics, Aarhus University.  He is interested in algorithms for big data, in particular, data stream algorithms, sublinear algorithms, algorithms on distributed data, I/O-efficient algorithms, and data structures.
\end{IEEEbiography}

\end{document}